\documentclass[%
reprint,
superscriptaddress,
amsmath,amssymb,
aps,
prx,
]{revtex4-2}

\usepackage{amsthm}

\usepackage{graphicx}
\usepackage{dcolumn}
\usepackage{bm}
\usepackage{siunitx}
\usepackage{physics}
\usepackage{hyperref}
\usepackage{csquotes}
\usepackage{tikz}

\usepackage{blochsphere}

\usepackage{ulem}

\newtheorem{theorem}{Theorem}
\newtheorem{conjecture}[theorem]{Conjecture}
\newtheorem{lemma}[theorem]{Lemma}

\begin{document}

\preprint{APS/123-QED}

\title{Robust certification of non-projective measurements: theory and experiment}

\author{Raphael Brinster}
\thanks{These authors contributed equally}
\affiliation{Institut für Theoretische Physik III, Heinrich-Heine-Universität Düsseldorf, Universitätsstraße 1, 40225 Düsseldorf, Germany}

\author{Peter Tirler}
\thanks{These authors contributed equally}
\affiliation{Universität Innsbruck, Institut für Experimentalphysik, Technikerstraße 25, 6020 Innsbruck, Austria}

\author{Shishir Khandelwal}
\affiliation{Physics Department and NanoLund, Lund University, Box 118, 22100 Lund, Sweden.}

\author{Michael Meth}
\affiliation{Universität Innsbruck, Institut für Experimentalphysik, Technikerstraße 25, 6020 Innsbruck, Austria}

\author{Hermann Kampermann}
\affiliation{Institut für Theoretische Physik III, Heinrich-Heine-Universität Düsseldorf, Universitätsstraße 1, 40225 Düsseldorf, Germany}

\author{Dagmar Bruß}
\affiliation{Institut für Theoretische Physik III, Heinrich-Heine-Universität Düsseldorf, Universitätsstraße 1, 40225 Düsseldorf, Germany}

\author{Rainer Blatt}
\affiliation{Universität Innsbruck, Institut für Experimentalphysik, Technikerstraße 25, 6020 Innsbruck, Austria}
\affiliation{Institute for Quantum Optics and Quantum Information of the Austrian Academy of Sciences, Technikerstraße 21a, 6020 Innsbruck, Austria}

\author{Martin Ringbauer}
\affiliation{Universität Innsbruck, Institut für Experimentalphysik, Technikerstraße 25, 6020 Innsbruck, Austria}

\author{Armin Tavakoli}
\affiliation{Physics Department and NanoLund, Lund University, Box 118, 22100 Lund, Sweden.}

\author{Nikolai Wyderka}
\affiliation{Institut für Theoretische Physik III, Heinrich-Heine-Universität Düsseldorf, Universitätsstraße 1, 40225 Düsseldorf, Germany}

\begin{abstract}
Determining the conditions under which positive operator-valued measures (POVMs), the most general class of quantum measurements, outperform projective measurements remains a challenging and largely unresolved problem.
Of particular interest are projectively simulable POVMs, which can be realized through probabilistic mixtures of projective measurements, and therefore offer no advantage over projective schemes. Characterizing the boundary between simulable and non-simulable POVMs is, however, a difficult task, and existing tools either fail to scale efficiently, provide limited experimental feasibility or work only for specific POVMs. Here, we introduce and demonstrate a general method to certify non-simulability of a POVM by introducing a complete hierarchy of semidefinite programs. It provides upper bounds on the non-simulability measure of critical visibility of arbitrary POVMs which are tight in many cases and outperform previously known criteria.
We experimentally certify the non-simulability of two- and three-dimensional POVMs using a trapped-ion qudit quantum processor by constructing non-simulability witnesses and introduce a modification of our framework that makes them robust against state preparation errors. Finally, we extend our results to the setting where an additional ancilla system is available. 

\end{abstract}

\maketitle

Quantum measurements serve as an interface between the classical and the quantum world and lie at the heart of quantum technologies such as quantum key distribution \cite{Pirandola_2020}, quantum metrology \cite{Giovannetti_2011}, and quantum computing \cite{Briegel_2009}. Usually, one considers projective measurements, which are represented by projectors onto the elements of an orthonormal measurement basis. However, for some tasks it is advantageous to employ generalized measurements, so-called positive operator-valued measures (POVMs). Although these cannot be realized directly due to their non-projective nature, they can be implemented as projective measurements in a higher-dimensional Hilbert space. While this need for additional dimensions complicates their implementation in experiments, they are known to yield advantages in tasks such as state discrimination~\cite{Barnett_2009} and state tomography~\cite{Scott2006, stricker_experimental_2022}. Despite these challenges, numerous experiments across different physical platforms have successfully implemented non-projective POVMs \cite{Bian2015,shahandeh2017ultrafine,Hou2018,Wang2023,Feng2025}

Given their value in these tasks, much work has been devoted to quantifying the dimension of the ancillary Hilbert space required to implement certain POVMs \cite{Oszmaniec2017,Kotowski2025,gomez2016,Martinez2023}. It has been shown that some POVMs, even if they are not projective, can be implemented within the original space through classical mixing of randomly selected projective measurements~\cite{Oszmaniec2017}. POVMs that can be simulated in such a way are called projectively simulable. They form a proper subset of all POVMs and yield no advantage over projective measurements.  

Despite their intuitive definition, not much is known about the set of simulable POVMs, other than that sufficient white noise will eventually make every POVM simulable \cite{Oszmaniec2017}.
A measure of simulability is thus given by the \textit{critical visibility}~\cite{Oszmaniec2017}, which quantifies the white noise robustness of POVMs with respect to becoming simulable. A critical visibility of $1$ signifies that a POVM is simulable, whereas smaller values indicate larger distances to the set of simulable POVMs. Notwithstanding recent progress \cite{Oszmaniec_2019, Kotowski2025, Cobucci2025}, several open questions and key issues remain, impeding a thorough understanding of the geometry of the set of simulable POVMs.

First, while it is known that POVMs become simulable if enough noise in terms of randomly selected outcomes is added, it is generally a hard task to quantify the required amount of noise exactly, as no efficient characterization of the simulable set in high dimensions is known. This is reminiscent of a similar situation in the context of robustness measures of entanglement~\cite{Chitambar_2019}. In a similar spirit, it is therefore desirable to find efficient outer approximations of the simulable set that allow us to find upper bounds on the visibility of POVMs. 

Second, we lack tools to calculate explicit decompositions of simulable POVMs in terms of projective measurements. Such tools would be of immediate experimental interest, as they provide a recipe for how to implement these POVMs without ancilla systems. 

Finally, implementing non-simulable POVM requires performing projective measurements in a larger dimensional Hilbert space. Therefore, demonstrating non-simulability acts as a certificate of control over additional degrees of freedom.
Such claims, however, need to be thoroughly and robustly certified, which, until now, have only succeeded for few, specific  measurements~\cite{gomez2016, Smania2020, Martinez2023}.

Here, we tackle these issues in a systematic manner. To this end, we introduce a complete hierarchy of semidefinite programs that yields a sequence of efficiently computable outer approximations of the set of simulable POVMs, which can be used to calculate upper bounds on the critical visibility of arbitrary POVMs. We provide numerical evidence that the hierarchy collapses at finite levels. We observe that in many cases, this collapse can be used to calculate exact visibilities together with specific decompositions of simulable POVMs. Furthermore, we show that our hierarchy outperforms recent similar approaches in two ways. First of all, we construct a specific family of non-simulable POVMs that, while being detected by our hierarchy, are not detected to be non-simulable by the previous detection methods proposed in \cite{Oszmaniec2017,Cobucci2025}. Second, exploiting the duality of semidefinite programs, we construct witnesses for non-simulability. Such witnesses can be measured experimentally by measuring the target POVM on a carefully chosen set of probe states. While such witnesses can be extracted from any SDP-based certification method, they usually suffer from being very susceptible to errors in state preparation. We overcome this limitation in our approach by modifying the hierarchy to significantly lower the demands on the fidelity of the prepared states for successful certification.

Next, we show the experimental viability of our methods by implementing a qubit symmetric informationally complete (SIC) POVM and a qutrit real space informationally complete POVM on a trapped-ion qudit quantum processor~\cite{ringbauer_universal_2022} and use the constructed witnesses to certify their non-simulability.

Finally, we extend our results to the setting where an additional ancilla of a fixed dimension is available. It turns out, that adding small ancillary systems can increase the projective simulability drastically \cite{Singal2022,Kotowski2025}. To complement the existing lower bounds in these scenarios, we develop tools to calculate corresponding upper bounds on simulability thresholds in the presence of ancillary systems.

The paper is organized as follows. After introducing projective simulability in Sec.~\ref{sec:SP}, we construct the complete hierarchy of semidefinite programs and provide evidence for its finite level collapse in Sec.~\ref{sec:SDP}. We apply it to a variety of POVMs to calculate upper bounds on the critical visibility and show that our hierarchy outperforms previous criteria. In Sec.~\ref{sec:SDPdual}, we construct non-simulability witnesses and show how to make them robust against state preparation errors in Sec.~\ref{sec:prep_errors}. In Sec.~\ref{sec:experiment}, we experimentally demonstrate the relevance of our methods on a trapped-ion quantum processor. Finally, in Sec.~\ref{sec:heralded}, we develop methods for bounding projective simulability thresholds in the presence of ancilla systems.

\section{Projectively simulable POVMs}\label{sec:SP}

A $d$-dimensional POVM with $n$ effects is a set of positive semidefinite matrices that sum to the identity, i.e.
\begin{align}
    \mathcal{M} = \{M_0,\ldots,M_{n-1}\}, \quad M_i\geq 0, \quad \sum_{i=0}^{n-1} M_i = \mathbb{I}_d.
\end{align}
We denote the set of $d$-dimensional $n$-effect POVMs by $\mathbb{M}(d,n)$. A POVM is called projective, if its elements are orthogonal projectors, i.e.
\begin{align}
    \mathbf{P} = \{P_0,\ldots,P_{n-1}\}, \quad P_i P_j = \delta_{ij}P_i, \quad \sum_{i=0}^{n-1} P_i = \mathbb{I}_d.
\end{align}

In Ref.~\cite{Oszmaniec2017}, the notion of a projectively simulable POVM was introduced and defined as a classical probabilistic mixture of projective measurements followed by post-processing, which is then proven to be irrelevant \cite{Oszmaniec2017, Cobucci2025}.
Mathematically, the set of all simulable POVMs, denoted by $\mathcal{P}(d,n)$, is therefore given by the convex hull of the set of projective measurements. The effects of a projectively simulable POVM can therefore be written as
\begin{align}
    M_i = \sum_{k=1}^N p_k P^k_i, 
\end{align}
where $\{ P_0^k, \dots, P_{n-1}^k \}$ are projective measurements for any $k$ and $\{ p_1,\dots,p_N \}$ forms a probability distribution.

The simplest quantitative measure describing the non-projectivity of a POVM $\mathcal{M}$ is the critical visibility $t(\mathcal{M})$. Adding enough white noise to each effect $M_i$ renders the noisy POVM defined via
\begin{align}
     \Phi_t(M_i) = t M_i + (1-t)\frac{\Tr M_i}{d} \mathbb{I}, \; 0\leq t \leq 1,
\end{align}
simulable. The critical visibility (sometimes called the simulability threshold) is then defined by
\begin{align}
    t(\mathcal{M}) := \max\{ t \;  | \;  \Phi_t(\mathcal{M}) \in \mathcal{P}(d,n), t \leq 1 \}.
\end{align}
For a noise parameter of $t = \max(\tfrac{1}{d}, 0.02)$, it is known that any $d$-dimensional POVM becomes simulable. Equivalently, one has $t(\mathcal{M}) \geq \max(\tfrac{1}{d}, 0.02)$ for every $d$-dimensional POVM $\mathcal{M}$~\cite{Oszmaniec2017, Kotowski2025}. At least for small dimensions this bound seems to be not tight, since for $d=2$ actually $t=\sqrt{\tfrac{2}{3}} \approx 0.8165$ is enough to simulate all POVMs by projective measurements \cite{Hirsch2017betterlocalhidden}. In Ref.~\cite{Oszmaniec2017}, efficient criteria based on semidefinite programs (SDPs) were introduced that can calculate the critical visibility simulability for POVMs in dimensions $d=2$ and $d=3$. They were later refined in Ref.~\cite{Cobucci2025}, but in most cases, they only yield upper bounds in $d\geq 4$.

In the the next section, we present a method to (in principle) calculate the critical visibility for given POVMs in any dimension.

\section{SDP Hierarchy} \label{sec:SDP}
We formulate the problem of certifying the non-simulability of a POVM in terms of a hierarchy of semidefinite programs~\cite{vandenberghe1996semidefinite}. We point out that our method is reminiscent of the symmetric extension for separable states, and the resulting Doherty-Parrilo-Spedalieri (DPS) hierarchy, see \cite{doherty2004complete}.
We start by introducing the main idea of the first non-trivial level of the hierarchy. To that end, consider a $d$-dimensional $n$-outcome simulable POVM $\mathcal{M} = \{ M_0, \dots, M_{n-1} \}$, i.e.,  we can decompose its elements as
\begin{align}\label{eq:sim_povm_decomposition}
    M_i = \sum_{k=1}^N p_k P^k_i, 
\end{align}
where the $\mathbf{P}^{k} = \{ P_0^k, \dots, P_{n-1}^k \}$ form projective measurements for any $k$. 
\begin{figure}
    \centering
    \includegraphics[width=1.0\linewidth]{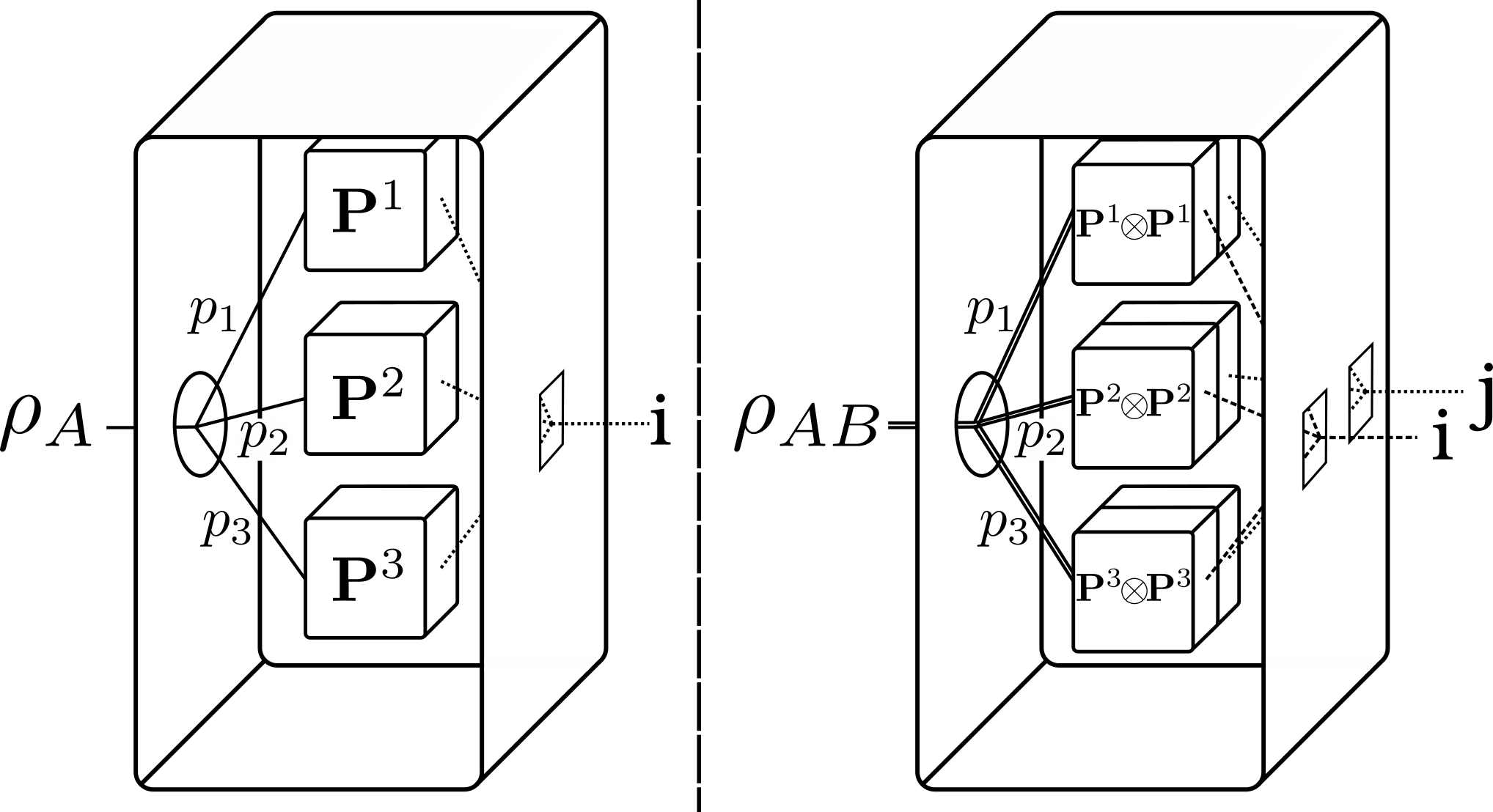}
    \caption{Left: a POVM is called simulable if it can be implemented by performing with probability $p_k$ the projective measurement $\mathbf{P}^k$. 
    Right: If a POVM on a single system is simulable, one can simulate as well a bipartite POVM by performing with probability $p_k$ the local projective measurement $\mathbf{P}^k \otimes \mathbf{P}^k$, which, upon marginalization, yields the measurement statistics of the single-body marginals.}
    \label{fig:simulation}
\end{figure}

If we can simulate $\mathcal{M}$ using the projective measurements $\mathbf{P}^k$, we can likewise simulate a bipartite POVM by implementing with probability $p_k$ the local projective measurement $\mathbf{P}^k \otimes \mathbf{P}^k = \{P_i^k \otimes P_j^k\}_{i,j=0}^{n-1}$ with the property that summing over the outcomes of the first (or the second) party yields the measurement statistics of the POVM $\mathcal{M}$ of the second (or the first) subsystem. The procedure is visualized in Fig.~\ref{fig:simulation}. Given the decomposition in Eq.~\eqref{eq:sim_povm_decomposition}, we can construct the  corresponding POVM elements via
\begin{equation}\label{eq:R_ij}
    R_{ij} = \sum_{k=1}^N p_k (P_i^k \otimes P_j^k),
\end{equation}
which obey
\begin{align}
    R_{ij} &\geq 0, \label{eq:Rij_positive} \\
    \sum_{i=0}^{n-1} R_{ij} &= \sum_k p_k \left( \Big( \sum_i P_i^k\Big) \otimes P_j^k \right) = \mathbb{I} \otimes M_j \label{eq:sum_i_R_ij}, \\
    \sum_{j=0}^{n-1} R_{ij} &= \sum_k p_k \left(  P_i^k \otimes \Big( \sum_j P_j^k\Big) \right) =  M_i \otimes \mathbb{I} \label{eq:sum_j_R_ij}, \\
    \Tr_{1} (V R_{ij}) &= \sum_k p_k P_i^k P_j^k = 0 \quad \quad \forall \; i \neq j .\label{eq:swap_constraint} 
\end{align}
Here, $\Tr_1$ denotes the partial trace w.r.t.~the first subsystem.
Eqs.~\eqref{eq:sum_i_R_ij} and \eqref{eq:sum_j_R_ij} just follow from the fact that the $P_i^k$ form a measurement and therefore $\sum_{i} P_{i}^k = \mathbb{I} \; \forall \; k $. $V$ is the two-qudit swap operator defined as
\begin{equation}\label{eq:swap}
    V = \sum_{i,j=0}^{d-1} \ket{ij}\bra{ji}
\end{equation}
and fulfills the well-known relation
\begin{align}
    \Tr_{1}(V(A\otimes B)) = AB.
\end{align}
Eq.~\eqref{eq:swap_constraint} then follows from the fact that elements of a projective measurement are orthogonal, i.e. $P_i^k P_j^k = 0 \; \forall \; k\; \text{and}\; i\neq j$. Note that together with the identity $\mathbb{I}$, the swap operator \eqref{eq:swap} forms a unitary $d^2$-dimensional representation of the symmetric group $S_2$ (\text{i.e.} all permutations acting on 2 symbols).

Constraints \eqref{eq:Rij_positive}-\eqref{eq:swap_constraint} can easily be generalized if we allow for larger $m$-partite matrices
\begin{align}\label{eq:R_i1im}
    R_{i_1\dots i_m} = \sum_{k=1}^{N}p_k ( P_{i_1}^k \otimes \dots \otimes P_{i_m}^k).
\end{align}
In order to generalize the swap constraint \eqref{eq:swap_constraint}, note that we can apply any swap operator, i.e.~all unitaries $V_\mu$ that act on $m$-partite qudit states via permutations (furnishing a $d^m$-dimensional unitary representation of the symmetric group $S_m$)
\begin{align}
    V_{\pi} = \sum_{i_1,\dots,i_m=0}^{d-1} \ketbra{\pi(i_1\dots i_m)}{i_1\dots i_m} \; \text{for} \;  \pi \in  S_m, 
\end{align}
where $\pi(i_1\dots i_m) = (i_{\pi^{-1}(1)}\ldots i_{\pi^{-1}(m)})$.
We define $\mathcal{V}_m$ to be the subset of all those unitaries that represent transpositions, meaning that exactly two indices are swapped. For example, in $m=3$ there are $3$ different swap operators corresponding to single transpositions 
\begin{align}
    \mathcal{V}_3 = \{V_{(12)},V_{(13)},V_{(23)}\},
\end{align}
where the subscript denotes the swapped systems. One possible swap constraint would then read as
\begin{align}
    \Tr_1 (V_{(12)}R_{i_1i_2i_3}) = \sum_k p_k (P_{i_1}^k \cdot P_{i_2}^k) \otimes P_{i_3}^k = 0 \; \forall \;  i_1 \neq i_2 .
\end{align}
In Appendix~\ref{app:primal_dual_sdp}, we demonstrate that it suffices to consider single transpositions only instead of all $m!$ permutations.

We can now define the following hierarchy of SDPs $H_m(\mathcal{M})$ with input POVM $\mathcal{M}$
\begin{widetext}
\begin{align}
  H_m(\mathcal{M}) = &\max_{R_{i_1\dots i_m}} \quad t & \label{sdp_multi_copy_primal_t(d)} \\
  \textrm{s.t. } & \quad \quad \quad \quad R_{i_1\dots i_m} \geq 0,& \label{sdp_multi_1} \\
  &\sum_{\substack{i_1,\ldots,i_{j-1},\\i_{j+1},\ldots,i_m=0}}^{n-1} \!\!\!\!\!R_{i_1\dots i_m} = \mathbb{I} \otimes \dots \otimes  \Phi_t(M_{i_j}) \otimes \dots \otimes \mathbb{I},  \quad \forall\, j=1\ldots m,\, i_j = 0\ldots n-1, \label{sdp_multi_sum}  \\
  &  \Tr_{a} (V_{(ab)} R_{i_1\dots i_m}) = 0, \quad \forall \; V_{(ab)} \in \mathcal{V}_m    \quad \text{if} \; V_{(ab)}\ket{i_1,\dots,i_m} \neq \ket{i_1,\dots,i_m} .  \label{sdp_multi_2}
\end{align}
\end{widetext}
Here, $\Tr_{a}$ means that we trace out system $a$, and we ensure that $i_a$ has to differ from $i_b$. This ensures that the swap constraint reduces to a product of distinct projectors and therefore vanishes. The SDP \eqref{sdp_multi_copy_primal_t(d)} takes a $d$-dimensional, $n$-outcome POVM $\mathcal{M}$ as input and calculates the minimal amount of noise (maximal $t$) one has to add, such that the necessary constraints for simulability (Eqs.~\eqref{sdp_multi_1}-\eqref{sdp_multi_2}) are fulfilled. Therefore, this yields an upper bound on $t(\mathcal{M})$. 

In general, higher levels of the hierarchy yield better upper bounds on $t(\mathcal{M})$, which is encapsulated in the following theorem.

\begin{theorem}[Higher level of the hierarchy leads to a better bound] \label{thm:theorem_1}
Let $t(\mathcal{M})$ be the critical visibility of a POVM $\mathcal{M}$ and $H_m(\mathcal{M})$ be the solution to the $m$'th level of the SDP hierarchy \eqref{sdp_multi_copy_primal_t(d)}. Then
\begin{align}\label{eq:theorem_1}
    H_2(\mathcal{M}) \geq H_3(\mathcal{M}) \geq \dots \geq t(\mathcal{M}).
\end{align}
\end{theorem}
\begin{proof}
Since $H_q(\mathcal{M})$ only consists of necessary conditions on projective simulability, and we are maximizing $t$, we trivially have
\begin{align}
    H_q(\mathcal{M}) \geq t(\mathcal{M}) \, \forall \; q \geq 2.
\end{align}
Now we need to show that from a solution $R_{i_1\dots i_{q+1}}$ to the level $m=q+1$, we can always construct a feasible point for the $q$-th level. For this, just define
\begin{align}\label{eq:Rq_from_q+1}
    R_{i_1\dots i_q} := \frac{1}{d} \sum_{i_{q+1}} \Tr_{q+1}R_{i_1\dots i_{q+1}}.
\end{align}
This obviously fulfills constraints \eqref{sdp_multi_1} and \eqref{sdp_multi_sum}. The swap constraint \eqref{sdp_multi_2} is fulfilled as well, because for a swap operator $V'_{(ab)} = V_{(ab)}'' \otimes \mathbb{I}^{q+1}$, that acts on system $q+1$ trivially, we have
\begin{align}
    \Tr_a \left( V_{(ab)}'' \otimes \mathbb{I}^{q+1} R_{i_1\dots i_{q+1}} \right) = 0\\
   \Rightarrow \Tr_{(ab)} \left( V_{(ab)}'' R_{i_1\dots i_q} \right) = 0 
\end{align}
Therefore, the choice from Eq.~\eqref{eq:Rq_from_q+1} fulfills all constraints in $H_q(\mathcal{M})$, and since we are maximizing the visibility $t$, we have
\begin{align}
    H_q(\mathcal{M}) \geq H_{q+1}(\mathcal{M}).
\end{align}
This proves the claim.
\end{proof}

In fact, one can strengthen Theorem~\ref{thm:theorem_1} and show that a slightly modified hierarchy is even complete, i.e., every non-simulable POVM is detected for large enough $m$:
\begin{theorem}[Completeness] \label{thm:completeness}
    Adding the constraints 
    \begin{align}
        V_{(ab)}R_{i_1\ldots i_a \ldots i_b \ldots i_m}V_{(ab)} &= R_{i_1\ldots i_b \ldots i_a \ldots i_m}, \label{eq:constraint_sym}\\
        \sum_{i_m=0}^{n-1} R_{i_1\ldots i_{m-1}i_m} &= \tilde{R}_{i_1\ldots i_{m-1}} \otimes \mathbb{I} \label{eq:constraint_sum}
    \end{align}
    for all $V_{(ab)} \in \mathcal{V}_m$, $i_j=0\ldots n-1$ and arbitrary positive semidefinite $\tilde{R}$ to the hierarchy $H_m$ in Eq.~\eqref{sdp_multi_copy_primal_t(d)} makes it complete.
\end{theorem}
The proof can be found in Appendix~\ref{app:proof_completeness}.
In practice, we found no difference in the obtained values of the original hierarchy and the complete one.

In the following, we benchmark our hierarchy on the following collection of POVMs. 
\begin{enumerate}
    \item \textit{SIC-POVMs}. 
A SIC-POVM (symmetric informationally complete POVM) in dimension $d$ consists of $d^2$ effects fulfilling $\Tr M_i M_j = (\delta_{ij}d + 1)/(d^2(d+1))$ and is known to be an optimal measurement for state tomography~\cite{Scott2006, stricker_experimental_2022}. While no general construction of a $d$-dimensional SIC-POVM is known, almost all known SIC-POVMs are constructed from a carefully chosen fiducial vector $\ket{\psi_0}$ such that $M_0 = \frac1d \ketbra{\psi_0}{\psi_0}$. The $d^2$ effects  are then constructed via
\begin{align}
    M_i = h_i M_0 h_i^\dagger.
\end{align}
Here $h_i$ are Weyl-Heisenberg operators~\cite{bengtsson2017}, i.e., $i=(i_1,i_2)$ is a multi-index and 
\begin{align}
    h_i = D_{i_1,i_2} = \tau^{i_1i_2} X^{i_1} Z^{i_2}, 
\end{align}
where $\tau = -e^{-\frac{i\pi}{d}}$,  $X=\sum_{j}^{d-1} \ketbra{j\oplus 1}{j}$ is the shift and $Z = \sum_{j=0}^{d-1} e^{2\pi i j /d} $ is the clock operator. 

The easiest case is that of $d=2$, where each SIC-POVM is unitarily equivalent to $\mathcal{M}_2 = \{M_0,M_1,M_2,M_3\}$, where $M_i = 1/2\; \ketbra{\Psi_i}{\Psi_i}$ with 
\begin{align}\label{eq:psi_sic2}
    \ket{\Psi_0} &= \ket{0},\nonumber\\
    \ket{\Psi_1} &= \frac{1}{\sqrt{3}}\ket{0} + \sqrt{\frac23}\ket{1},\nonumber\\
    \ket{\Psi_2} &= \frac{1}{\sqrt{3}}\ket{0} + e^{2\pi i /3}\sqrt{\frac23}\ket{1},\nonumber\\
    \ket{\Psi_3} &= \frac{1}{\sqrt{3}}\ket{0} + e^{4\pi i /3}\sqrt{\frac23}\ket{1}.
\end{align}

SIC-POVMs were conjectured to be most robust to projective simulability in Ref.~\cite{Oszmaniec2017}, which recently was proven to be false beyond dimension two~\cite{Cobucci2025}. In fact, in $d=3$, all SIC-POVMs are (anti-)unitarily equivalent to members of a single-parameter family of SIC-POVMs with different critical visibilities, constructed from the fiducial vector~\cite{Szollosi2014, Hughston2016, Zhu2010}
\begin{align}\label{eq:sic3fid}
    \ket{\psi_0(\varphi)} = \frac{1}{\sqrt{2}}(\ket{1} - e^{i\varphi}\ket{2}),
\end{align}
of which the most robust one is the so-called Hesse SIC with $\varphi = 0$ and a critical visibility of $[1+4\cos(\pi/9)]/6$ (the Hesse SIC is labeled SIC3c in Ref.~\cite{Scott_2010}. There, the cases of SIC3a and SIC3b correspond to choosing $\varphi = \pi/18$ and $\varphi = \pi/9$, respectively). SIC-POVMs constructed from $\varphi \notin [0,\pi/9]$ can be (anti-)unitarily mapped to a representative with $\varphi$ in that range.
In general, SIC-POVMs are known for finitely many dimensions, including every dimension up to $d=193$ \cite{Horodecki2022}, and they are conjectured to exist in all dimensions \cite{Zauner1999}.

\item \textit{Real space IC-POVM in $d=3$}.
If one restricts to the real subspace in $d=3$, one can define a real space informationally complete POVM instead, consisting of $\binom{d+1}{2}$ effects. Here, we choose $\mathcal{M}_{3r} = \{M_0,\ldots,M_5\}$ with the six effects $M_i = 1/2\;\ketbra{\Phi_i}{\Phi_i}$ with 
\begin{align}\label{eq:phi_sicr3}
    \ket{\Phi_{0,1}} &= \frac1{\sqrt{2}}(\ket{0} \pm \ket{1}),\nonumber\\
    \ket{\Phi_{2,3}} &= \frac1{\sqrt{2}}(\ket{0} \pm \ket{2}),\nonumber\\
    \ket{\Phi_{4,5}} &= \frac1{\sqrt{2}}(\ket{1} \pm \ket{2}).
\end{align}

\item \textit{General POVM in $d=3$}.
One can generalize the fiducial vector in Eq.~\eqref{eq:sic3fid} to an arbitrary vector
\begin{align}\label{eq:fid_vec_two_parameters}\ket{\phi_0(\vartheta, \varphi)} = \cos\frac{\vartheta}{2} \ket{1} + \sin \frac{\vartheta}{2} e^{i\phi}\ket{2}.
\end{align}
In this way, one can represent each member of this family by a point on a sphere with SIC-POVMs located on the equator at $\vartheta = \pi/2$ (note that $\varphi = \phi + \pi$). At the poles, e.g., where $\vartheta = 0$, the POVM is trivially simulable, implying that their critical visibility is equal to $1$.

\item \textit{Flagged SIC-POVMs}.
In Ref.~\cite{Cobucci2025}, it was found that there is a POVM in $d=4$ whose non-simulability is more robust to noise than that of any SIC-POVM. It is constructed by embedding the effects of the three-dimensional Hesse SIC into a four-dimensional space and adding the projector $\ketbra{3}{3}$ as an additional effect to obtain a complete POVM. The same principle can be used to create a flagged version of any POVM, and we additionally consider a three-dimensional, flagged SIC2 with five effects.

\item \textit{Two-copy SIC-POVM.}
Finally, we consider the four-dimensional POVM that is constructed by tensoring the four effects of the qubit SIC-POVM with the identity, i.e., we set $M_i = 1/2\; \ketbra{\Psi_i}{\Psi_i} \otimes \mathbb{I}$, where the $\ket{\Psi_i}$ are given in Eq.~\eqref{eq:psi_sic2}.
The SDP criterion from Ref.~\cite{Cobucci2025} does not detect the non-simulability of this POVM, however, as detailed below, our hierarchy does.

\end{enumerate}

\begin{table}
    \begin{tabular}{ l | c | c | c | c | r }
      $d$ & POVM $\mathcal{M}$ & $H_{2}(\mathcal{M})$ & $H_{3}(\mathcal{M})$ & $H_{4}(\mathcal{M})$ & $t(\mathcal{M})$ \\
      \hline
      2&$\text{SIC2}^\ast$  &0.8165 &0.8165 & 0.8165 & $\sqrt{2/3}$   \\ 
      \hline
      3&$\text{flag SIC2}^\ast$ &0.8193 & 0.7985 & 0.7985 &  0.7985 \\
      3&$\text{real IC3}^\ast$ & 0.8529 & 0.8521 & 0.8521 & 0.8521 \\
      3&$\text{SIC3a}^\ast$  &0.8334 &0.8004 &0.8004 & 0.8004  \\
      3&$\text{SIC3b}^\ast$  &0.8334 &0.8058 &0.8058 & 0.8058 \\
      3&$\text{SIC3c}^\ast$  &0.8334 &0.7932 &0.7932 & 0.7932\\
      \hline
      4&$\text{flag SIC3c}^\dagger$ & 0.8348 & 0.8002 & ? & 0.7824 \\
      4& $ \text{SIC2} \otimes \mathbb{I}_2^\sharp$ & 0.9443 & 0.8665 & 0.8661 & ?  \\
      4&$\text{SIC4a}^\dagger$  &0.8453 &0.8347 & ? & 0.8255 \\
      \hline
      5&SIC5a  &0.8544 &0.8461 &? &?  \\
      \hline
      6&SIC6a  &0.8617 &?      &? &? \\
      \hline
      7&SIC7a  &0.8676 &?      &? &? \\
      7&SIC7b  &0.8676 &?      &? &?  
    \end{tabular}
    \caption{SDP results for critical visibilities $t(\mathcal{M})$ in various dimensions using different levels $m$ of the hierarchy $H_m$ in Eq.~\eqref{sdp_multi_copy_primal_t(d)}. The labels for the SIC-POVMs correspond to those in Ref.~\cite{Scott_2010}. The $\ast$ denotes POVMs, where we get $t(\mathcal{M})$ by reading off a projective decomposition from the solution of $H_d(\mathcal{M})$ (see Appendix~\ref{app:decomp_primal}). The $\dagger$ denotes POVMs, where $t(\mathcal{M})$ was calculated in \cite{Cobucci2025}. $\sharp$ denotes the POVM, where the SDP from \cite{Cobucci2025} does not detect non-simulability, but $H_m(\mathcal{M})$ does.  }\label{table_t(d)}
\end{table}

We test our hierarchy on the introduced POVMs to benchmark our method and list the results in Table~\ref{table_t(d)}. Notably, our hierarchy certifies the non-simulability of the two-copy SIC-POVM $\mathcal{M}_2 \otimes \mathbb{I}$, giving the upper bound of $t(\mathcal{M}_2 \otimes \mathbb{I}) \leq 0.8661$, whereas the criterion of Ref.~\cite{Cobucci2025} and \cite{Oszmaniec2017} fails to detect its non-simulability completely. In fact, we find a whole family of POVMs, for which this is the case.
\begin{lemma}\label{lemma}
    Let $M = \{M_i\}_{i=0}^{n-1}$ be a POVM in dimension $d$ with rank-1 effects, $Tr(M_i) = \frac{d}{n}$ and $M_iM_j \neq 0$. Then define the $d\cdot m$-dimensional POVM $\tilde{M}$ with $\tilde{M}_i = M_i \otimes \mathbb{I}_m $, where $m$ is chosen, such that $m\frac{d}{n} \in \mathbb{N}$. Then $\tilde{M}$ is
    \begin{itemize}
        \item[] i) detected to be non-simulable by the first non-trivial level of our hierarchy, i.e., $H_2(\tilde{M}) < 1$,
        \item[] ii) not detected by the SDP criteria in \cite{Oszmaniec2017} and \cite{Cobucci2025}.
    \end{itemize}
\end{lemma}
The proof can be found in Appendix \ref{app:proof_lemma}. This shows that our hierarchy is stronger than previous criteria not only in the limit of large levels, but already at its first non-trivial level.

For the two-parameter family of POVMs constructed from the vectors in Eq.~\eqref{eq:fid_vec_two_parameters}, we show the results in Fig.~\ref{fig:lemonsqueezer}, where each point on the outer sphere is matched by a point on the inner spheroid-like convex set, which is constructed such that its distance to the sphere corresponds to the critical visibility.

We note that in many cases we studied, the solution $R_{i_1\dots i_m}$ directly reveals the decomposition of the noisy POVM $\Phi_t(\mathcal{M})$  into projective measurements \eqref{eq:sim_povm_decomposition}. As a result, the computed critical visibility $t(\mathcal{M})$ is tight. For the cases we studied, this seems to happen at the $d$-th order of the SDP hierarchy for a $d$-dimensional POVM. It is not obvious that this should be possible, since from \eqref{eq:R_i1im} it seems that all $R_{i_1,\dots,i_m}$ are non-trivial sums of scaled tensor products of projectors and it seems difficult to extract a single projector $P_{i_1}^k$ or the right probability $p_k$, see Appendix~\ref{app:decomp_primal} for details. The fact that we get explicit decompositions of SIC-POVMs as well as random rank-1 POVMs in $d=2,3$ leads us to the following conjecture.

\begin{conjecture} For a $d$-dimensional POVM $\mathcal{M}$ the hierarchy $H_m(\mathcal{M})$ defined in \eqref{sdp_multi_copy_primal_t(d)} collapses at the $d$'th level, at which point the solution equals the critical visibility 
\begin{align}
    t(\mathcal{M}) = H_d(\mathcal{M}).
\end{align}
\end{conjecture}
Unfortunately, we are not able to prove this conjecture or get more evidence in higher dimensions, since $H_d(\mathcal{M})$ scales exponentially with $d$. However, we find that it holds true, for instance, for random rank-1 POVMs in $d=2,3$ and for the whole two-parameter family of POVMs constructed from the fiducial vector \eqref{eq:fid_vec_two_parameters}, implying that the representation of the simulable set in Fig.~\ref{fig:lemonsqueezer} is tight, yielding a faithful representation of the set of simulable POVMs in this family.

\begin{figure}
    \centering
    \includegraphics[width=0.99\columnwidth]{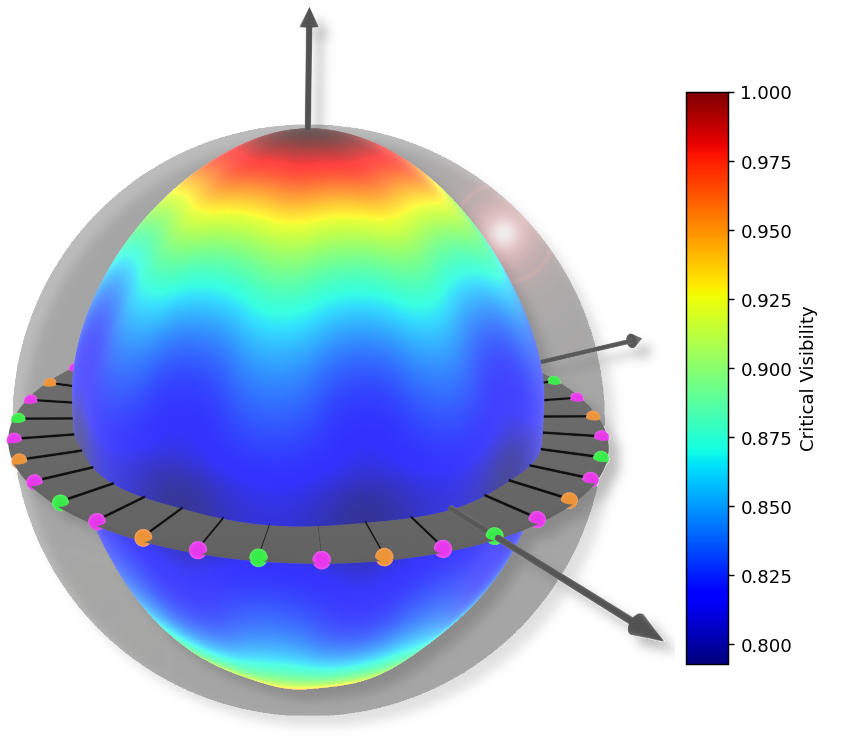}
    \caption{Faithful representation of the set of POVMs constructed from the fiducial vector in Eq.~\eqref{eq:fid_vec_two_parameters}, where $\vartheta$ denotes the polar and $\phi$ the azimuthal angle. Each point on the transparent unit sphere corresponds to the fiducial vector of one of the POVMs. The SIC-POVMs, constructed from Eq.~\eqref{eq:sic3fid}), and their (anti-)unitary equivalents are located on the equator. The origin corresponds to the maximally noisy measurement, such that noisy versions $\Phi_t(M)$ of POVM $M$ on the surface of the sphere lie on a straight line to the origin. The surface of the inner convex set corresponds then to the noisy POVMs at their critical visibilities, which in this picture corresponds to the distance to the origin. This distance is highlighted by the color coding.
    On the equator, we highlight some special choices of equivalence classes of SIC-POVMs which correspond to those having lowest (Hesse SIC, orange, corresponding to the choice of $\phi = \pi$, i.e., $\varphi = 0$ in Eq.~\eqref{eq:sic3fid}), highest (green) and intermediate (purple) critical visibility. 
    \label{fig:lemonsqueezer}}
\end{figure}

We conclude this section by highlighting another application of our hierarchy in the context of state discrimination. To that end, we ask for how much arbitrary noise in form of another POVM can a measurement tolerate before eventually becoming simulable. To that end, we consider the noise model   
\begin{align}
    \tilde{\Phi}_{t,\mathcal{N}}(M_i) = tM_i + (1-t) N_i 
\end{align}
for some arbitrary POVM $\mathcal{N} = \{N_0,\ldots,N_{n-1}\}$. This allows to modify the hierarchy $H_m$ in Eq.~\eqref{sdp_multi_copy_primal_t(d)} by additionally maximizing over $\mathcal{N}$ and replacing $\Phi_t$ by $\tilde{\Phi}_{t,\mathcal{N}}$ in Eq.~\eqref{sdp_multi_sum} (technically, we optimize over subnormalized POVMs that sum to $(1-t)\mathbb{I}$). In this way, the hierarchy yields values $\tilde{H}_m(\mathcal{M})$ which act as lower bounds on the so-called generalized robustness $R$ via $R(\mathcal{M}) \geq 1/\tilde{H}_m(\mathcal{M}) - 1$
\cite{Chitambar_2019, steiner2003generalized}. Due to completeness of our hierarchy, the bound becomes an equality in the limit of large $m$. The relevance of this bound lies in the fact that the generalized robustness precisely quantifies the achievable advantage of using non-simulable measurements over simulable ones in a state discrimination task  \cite{oszmaniec2019operational}.

\section{Dual SDP and non-simulability witnesses}\label{sec:SDPdual}

By studying the dual SDP we can construct non-simulability witnesses. The dual problem of $H_2(\mathcal{M})$ is calculated in Appendix~\ref{app:primal_dual_sdp} and takes the following form ($V$ is again the swap operator from equation \eqref{eq:swap})
\begin{align}\label{sdp_dual_t(d)}
\mathcal{D}_2(\mathcal{M})=& \min_{A^i,D^{ij}} \quad  \sum_{i=0}^{n-1} \frac{2}{d} \left[\Tr(A^i) \Tr(M_i)  \right] \\ 
\textrm{s.t.} \quad & A^i + V A^j V +  (\mathbb{I}\otimes D^{ij} V + \text{h.c.}) \geq 0 \; \forall \; i,j  \;, \label{sdp_dual_pos} \\
&\text{with} \; D^{ij}=(D^{ji})^\dagger, D^{ii} = 0 , \nonumber  \\
& \sum_{i=0}^{n-1} 2 \left[\Tr(A^i) \frac{\Tr(M_i)}{d} - \Tr(\Tr_2 (A^i) M_i)\right]  = 1. \label{sdp_dual_normalisation}
\end{align}
The $A^i$ are hermitian $d^2 \times d^2$ and $D^{ij}$ complex $d \times d$ matrices. Calculating the critical visibility directly via the dual side also yields advantages, as symmetries can be used to reduce the number of optimization variables and constraints, as detailed in Appendix~\ref{app:primal_dual_sdp}.

Similar to the case of entanglement witnesses, we can construct non-simulability witnesses from the dual SDP.

\begin{theorem}\label{thm:witnesspositivity}
Let $A^0, \ldots, A^{n-1}$ be operators satisfying equation~\eqref{sdp_dual_pos} and define $W_i = \Tr_2 A^i$. Then for any simulable POVM $(M_0,\ldots,M_{n-1}) \in \mathcal{P}(d,n)$
\begin{align}\label{eq:witness}
    \expval{M_i} := \sum_{i=0}^{n-1} \Tr(W_i M_i) \geq 0.
\end{align}
\end{theorem}
This can be shown by applying the weak duality of semidefinite programming, for a more detailed proof, see Appendix \ref{app:proof_theorem}.

Conversely, if we find a POVM where Eq.~\eqref{eq:witness} becomes negative, we know that the solution is not in agreement with $t\geq1$ and the POVM is therefore not projectively simulable. Eq.~\eqref{eq:witness} hence describes a non-simulability witness (see Fig. \ref{fig:witness}).

In principle, the witness can be evaluated directly using measurement statistics, based on the spectral decomposition of the witness operators
\begin{align}\label{eq:witness_spectral}
    W_i = \sum_{k=0}^{d-1} \lambda_i^k \ketbra{\lambda_i^k}{\lambda_i^k} .
\end{align}
Eq.~\eqref{eq:witness} then reads
\begin{align}\label{eq:witness_qik}
    \sum_{i=0}^{n-1} \left[\Tr(W_i M_i) \right] = \sum_{i=0}^{n-1} \sum_{k=0}^{d^2-1} \lambda_i^k \underbrace{\Tr ( \ketbra{\lambda_i^k}{\lambda_i^k}M_i)}_{q^k_{i}} ,
\end{align}
where $q_i^k$ denotes the probability of obtaining outcome $i$ upon performing the POVM measurement $\{M_0,\ldots,M_{n-1}\}$ on the input state $\ket{\lambda_i^k}$. A negative witness expectation value then certifies projective non-simulability of the performed measurement.

\section{Preparation Errors}\label{sec:prep_errors}
A major experimental challenge in evaluating Eq.~\eqref{eq:witness_qik} experimentally occurs if the states $\ket{\lambda_i^k}$ are not prepared perfectly. Similarly to the case of entanglement witnesses, where the same effect occurs if the \textit{measurement} is not implemented perfectly, this quickly leads to false detections of non-simulability \cite{morelli2022,tavakoli2021semi}. 

\begin{figure}
\centering
    \includegraphics[width=0.99\columnwidth]{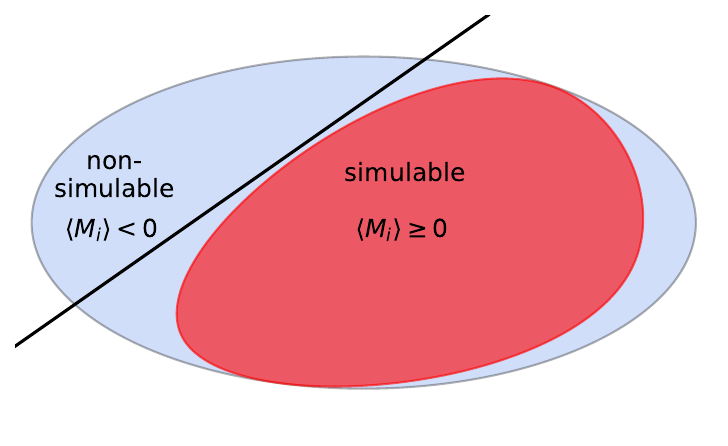}
    \caption{Schematic representation of the set of projectively simulable POVMs (red region) as a convex subset of the set of all POVMs. A non-simulability witness (solid line) separates simulable POVMs from non-simulable POVMs (blue region). For the definition of 
    $\expval{M_i}$ see Eq.~\eqref{eq:witness}. }
    \label{fig:witness}
\end{figure}

To quantify this, assume that we prepare states  $\rho_i^k$ that are close to $\ketbra{\lambda_i^k}{\lambda_i^k}$ in terms of the fidelity, such that the experimentally observed quantities are instead given by $p_i^k$ with
\begin{align}\label{eq:exp_prob_fid}
    p_i^k = \Tr(\rho_i^k M_i), \quad F(\rho_i^k,\ketbra{\lambda_i^k}{\lambda_i^k}) \equiv \bra{\lambda_i^k} \rho_i^k \ket{\lambda_i^k} \geq F_i^k.
\end{align}
Here, the $F_i^k$ are lower bounds on the fidelity of the prepared states w.r.t.~the target states.

In order to evaluate \eqref{eq:witness_qik}, one needs bounds on the $q_i^k$ depending on the measured probabilities $p_i^k$. Although this is possible using wasteful fidelity estimates (see, e.g., Ref.~\cite{morelli2022}), such estimates lead to very high demands on the fidelity on the order of 99.9\% in order to observe a negative expectation value of the witness. 

To circumvent this, we develop a different approach, yielding more relaxed bounds on the required fidelity and therefore making the experimental certification much more robust.  The main idea is the following: we measure the POVM $(M_0,\ldots,M_{n-1})$ on approximate witness states $\rho_\alpha$ assuming $ F(\rho_\alpha,\ketbra{\lambda_\alpha}{\lambda_\alpha}) \geq F_\alpha$, and then check if the measured probabilities on the ensemble of states are compatible with a projectively simulable POVM using a suitable linear semidefinite approximation of a joint space of measurements and prepared states, thereby eliminating the need for overly tight fidelity bounds. We extend our SDP hierarchy to include experimental data and expand the $R$-matrices by an additional system to encapsulate the (unknown) experimentally prepared states. An infeasible outcome then implies non-simulability of the performed POVM. For this, we define
\begin{align}
    R_{ij\alpha} = \sum_\lambda p_\lambda  P_{i}^\lambda \otimes P_{j}^\lambda \otimes \rho_\alpha.
\end{align}
These $R_{ij\alpha}$ are optimization variables in the SDP, which fulfill similar constraints to those in \eqref{sdp_multi_copy_primal_t(d)} and furthermore can be related to the measured probabilities and fidelities in Eq.~\eqref{eq:exp_prob_fid}. For example,
\begin{align}
   \Tr_1 \sum_{j} R_{ij\alpha} &= \Tr (M_i) \mathbb{I} \otimes \rho_\alpha \geq \lambda_i^{\text{max}} \mathbb{I} \otimes \rho_\alpha \nonumber \\
  &\geq p^{\text{max}}_{i} \, \mathbb{I} \otimes \rho_\alpha \, \forall \alpha, i,
\end{align}
where $ \lambda_i^{\text{max}}$ is the maximum eigenvalue of $M_i$ and $p^{\text{max}}_{i}$ the maximal probability measured for the outcome $i$. Additionally, the $\rho_\alpha$ are optimization variables constrained by the fidelity estimates via
\begin{align}\label{eq:fidelity_constraint}
    \Tr (\rho_\alpha\ketbra{\lambda_\alpha}{\lambda_\alpha}) \geq F_\alpha.
\end{align}
This method also takes statistical errors $\delta$ of the probabilities $p_{\alpha,i}$ into account, from which we can deduce a confidence interval via Hoeffding's bound \cite{Hoeffding01031963}
\begin{align}\label{eq:hoeffding}
    \mathbb{P}\left( |p_{\alpha,i} - \mathbb{E}[\overline{p_{\alpha,i}}]| \geq \delta \right) \leq 2 \exp\left( -2 N \delta^2 \right),
\end{align}
where $ \mathbb{E}[\overline{p_{\alpha,i}}]$ is the expected value of $p_{\alpha,i}$ and $N$ the number of shots in the experiment. The complete extended SDP is listed in Eq.~\eqref{sdp_cert} in Appendix~\ref{app:certify_exp}. It is  a hybrid version of the primal and dual one, because it optimizes over some $R_{ij\alpha}$ matrices, which have very similar constraints as the $R_{ij}$ in the primal SDP \eqref{sdp_multi_copy_primal_t(d)} and also states $\rho_\alpha$ which are close to the witness states from the dual side. In principle, the states are not restricted to be the witness states and can instead be arbitrary. However, the constraints in the SDP might not be strong enough for differently chosen states, and we observe that, usually, preparing eigenvectors of the witness operators leads to the optimal expectation value of the witness.

Finally, we note that also the criteria from Refs.~\cite{Oszmaniec2017, Cobucci2025} can be used to obtain non-simulability witnesses via their duals. However, the method to account for preparation errors introduced here is not applicable to them and we are not aware of a method to make those witnesses robust.

\section{Experimental certification}\label{sec:experiment}
To demonstrate the practicality of our witnesses, we measured them for two POVMs. The first one is the qubit SIC-POVM $\mathcal{M}_2$ from Eq.~\eqref{eq:psi_sic2}, which is known to be the most non-projective POVM in dimension two~\cite{Hirsch2017betterlocalhidden} and has relevance for applications such as quantum state tomography~\cite{stricker_experimental_2022}. The second POVM measured is the real space IC-POVM in $d=3$, $\mathcal{M}_{3r}$ from Eq.~\eqref{eq:phi_sicr3}.  

Using the second level of the dual SDP \eqref{sdp_dual_t(d)}, we obtain that the optimal witness for both POVMs is given by measuring directly the effect vectors \eqref{eq:psi_sic2} and \eqref{eq:phi_sicr3}, respectively. In particular, for the 2d SIC-POVM, we get
\begin{align}
    W_i = \lambda_- \ketbra{\Psi_i}{\Psi_i} + \lambda_+ (\mathbb{I} - \ketbra{\Psi_i}{\Psi_i})
\end{align} with $\lambda_\pm = 1/\sqrt{24} \pm 1/4$. For $\mathcal{M}_{3r}$, we obtain
\begin{align}
    W_0 = \kappa_0 \ketbra{\Phi_0}{\Phi_0} + \kappa_1 \ketbra{\Phi_1}{\Phi_1} + \kappa_2 \ketbra{2}{2}
\end{align}
with $\kappa_0 \approx -0.02453$, $\kappa_1 \approx 0.2287$ and $\kappa_2 \approx 0.2223$. The other $W_i$ are obtained from $W_0$ upon conjugation with elements from the symmetry group generated from $X = \ketbra{1}{0} + \ketbra{2}{1} + \ketbra{0}{2}$ and $Z_3 = \ketbra{0}{0} + \ketbra{1}{1} - \ketbra{2}{2}$.

\begin{figure}
    \centering
    \includegraphics[width=1.0\columnwidth]{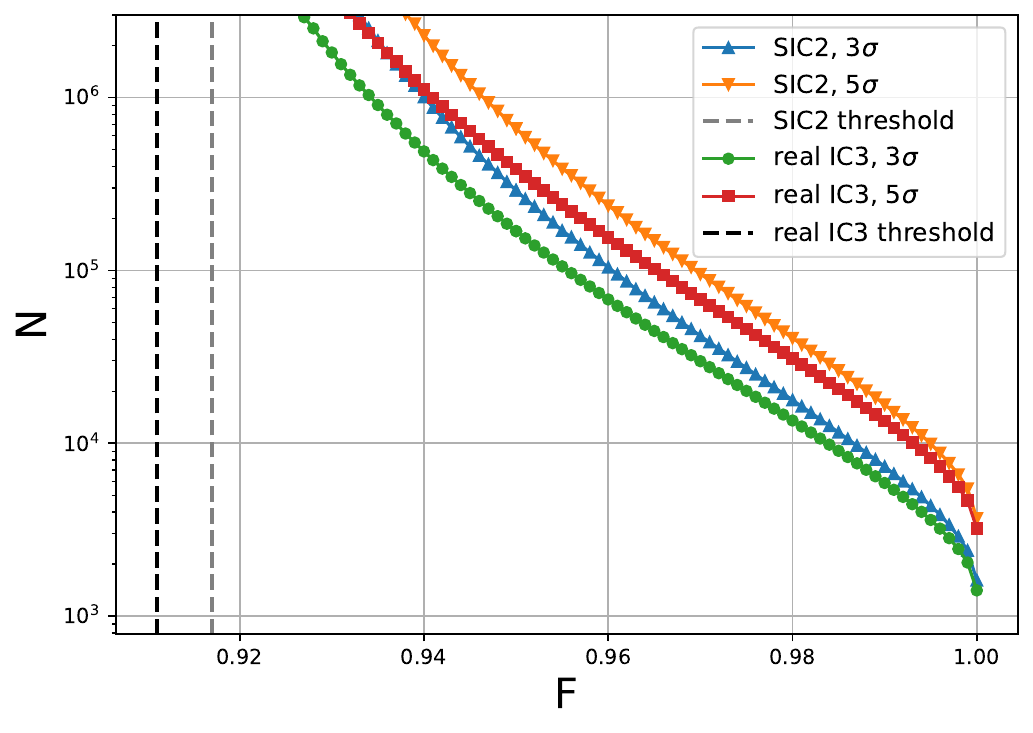}
    \caption{Number of shots $N$ versus the required fidelity $F$ of the probe states for the qubit SIC and real-space qutrit IC3 POVM, to obtain a violation with confidence intervals of $3\sigma$ and $5\sigma$. Thresholds correspond to fidelities at which no certification is possible anymore. \label{fig:N_vs_F}}
\end{figure}
To get an estimate of the state preparation fidelity \eqref{eq:fidelity_constraint} and the required number of shots $N$ per state, we set the probabilities to the ideally measured ones $p_{\alpha,i} = \Tr(\rho_\alpha M_i) $ and calculate the minimal statistical error $\delta$ using Eq.~\eqref{sdp_cert}. Then Hoeffding's bound in Eq.~\eqref{eq:hoeffding} can be used to calculate the minimum number of shots to achieve a certain confidence interval. In Fig.~\ref{fig:N_vs_F}, the minimum number of shots at given state fidelities for $3\sigma$ and $5\sigma$ certification intervals are plotted.

Both POVMs are measured on a universal qudit quantum processor using one trapped  $^{40}\text{Ca}^+$ ion in a macroscopic linear Paul trap. The trap is designed to hold linear strings of ions, confined through a combination of an rf-potential oscillating at \SI{23.5}{\mega \hertz} and a static potential generating confinement along the crystal axis. A sketch of the trap geometry can be seen in Fig.~\ref{fig:trap}. The logical state of the ion is encoded in its internal energy levels, and can be read out using state-dependent scattering of light that is collected on an EMCCD camera. State preparation is implemented using optical pumping as well as Doppler-, Polarization-gradient- and Resolved Sideband Cooling. A thorough description of the experimental setup can be found in Appendix~\ref{app:setup}, as well as in~Ref.~\cite{ringbauer_universal_2022}.

We used the same processor for measuring both the 2d SIC-POVM $\mathcal{M}_2$ in a qubit and the 3d IC3-POVM $\mathcal{M}_{3}$ in a qutrit. In $^{40}\text{Ca}^+$, we can access up to eight (meta)stable states suitable for quantum logic, of which seven can be experimentally distinguished. Qubit and qutrit are encoded into two or three of the available levels respectively. The measurement of the qubit SIC (qutrit IC3) POVM is then realized using a Naimark dilation that maps the information onto 4 (6) of the available levels, followed by a projective qudit readout. All logical operations were carried out using resonant laser pulses coupling to the $S_{\frac{1}{2}} \to D_{\frac{5}{2}}$ transitions. The allocation of qudit states to energy levels is shown in Fig.~\ref{fig:level} in Appendix~\ref{app:setup}.

\begin{figure}
    \centering
    \includegraphics[width=1.0\linewidth]{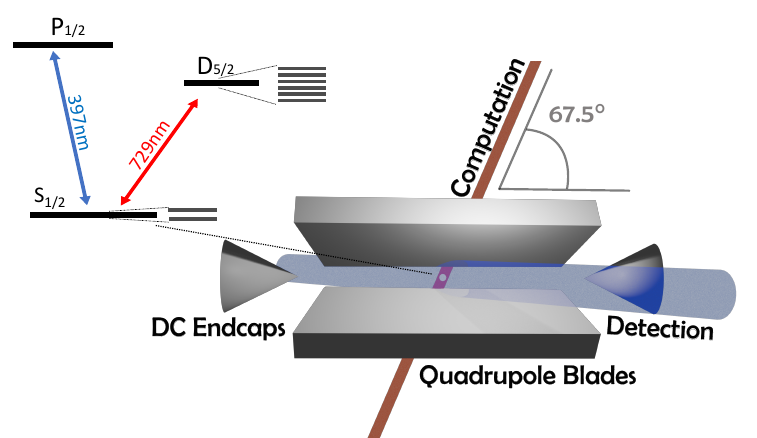}
    \caption{Schematic view of the trap and the course level structure of the $^{40}\text{Ca}^+$ ion. A combination of RF-potential on the quadrupole blades and a DC potential on the endcaps confines the ion in the trap centre. Detection and readout are performed using a \SI{397}{\nano \meter} laser, computation or state manipulation using a tightly focused \SI{729}{\nano \meter} laser, coming in at a 67.5 $^\circ$ angle to the trap axis.}
    \label{fig:trap}
\end{figure}

The certification consists of two steps. The first is the certification of the state preparation fidelity on a set of fiducial states, a lower bound which is required as an input for the experimentally constrained SDP, in particular the constraint in Eq.~\eqref{eq:fidelity_constraint}. Then, the POVM is measured on this set of states.

The state certification protocol consists of preparing the desired state, and then measuring in the eigenbasis of the target state. This procedure, which can be understood as a variant of direct fidelity estimation~\cite{flammia_direct_fidelity_theory}, directly gives access to the overlap or fidelity between the prepared and desired state. Notably, the measurement bases are chosen in a way that coherent errors in the state preparation are not cancelled by, but rather add to errors in the measurement rotations. The fidelity estimate thus includes the state preparation and measurement errors and is guaranteed to be a lower bound. Every state was certified with 50.000 shots. This was a compromise between measurement time and data quantity, since more shots reduce the threshold fidelity needed to certify non-simulabilty, see Fig.~\ref{fig:N_vs_F}. The results are shown in Fig.~\ref{fig:fidelities}. The corresponding numerical values can be found in Tab.~\ref{tab:fidelities} in Appendix~\ref{app:POVMmeasurement}. 

In the second step of the protocol the POVM was then measured using the same state preparation sequence followed by a unitary implementing its Naimark dilation, such that it can be projectively measured in the extended space of a $4$-dimensional (in case of the SIC-POVM) or a $6$-dimensional (in case of the real-space IC3 POVM) qudit. For each input state, 40.000 shots were used. For more details, see Appendix~\ref{app:experimental_imp}. The $p_i^k$ then correspond to the probabilities of the different measurement outcomes of the qudit after having applied the Naimark dilation and can be found in Tab.~\ref{tab:povm_statistics}.

The measured probabilities $p_i^k$ are then the basis of an SDP that evaluates the likelihood of the measured data to have come from a projectively simulable POVM. This way, we are able to certify with well over $5\sigma$ confidence (see Fig.~\ref{fig:N_vs_F}) that the performed qubit POVM $\mathcal{M}_2$ as well as the qutrit POVM $\mathcal{M}_{3r}$ could not have been simulated by projective measurements.

\begin{figure}
    \centering
    \includegraphics[width=\linewidth]{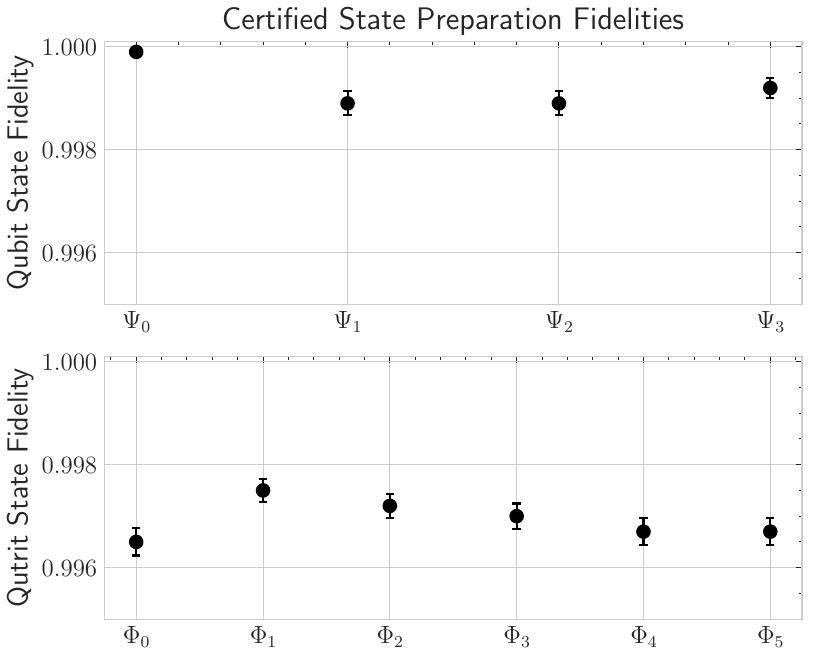}
    \caption{Certified fidelity for the different states to be measured by the POVMs. The top plot shows the states used for certifying the 2d-POVM, the bottom one the states used for the 3d POVM certification. Each state represents the average over 50000 shots.}
    \label{fig:fidelities}
\end{figure}

\section{POVM simulation with an ancilla}\label{sec:heralded}
In the preceding sections, we have discussed projective simulation for a situation where the experimenter has access to and control over only the $d$-dimensional system upon which the target POVM $\mathcal M$ acts. The situation changes considerably if a $d_A$-dimensional ancilla system is added to the picture, i.e., if the now $d_A d$-dimensional operators $P^k_i$ act on the joint system-ancilla Hilbert space, while the POVM is only $d$-dimensional. Projective simulation with an ancilla entails that the simulation gives the same measurement statistics as the target POVM. As before, the projective simulability of a POVM can be characterized by the critical visibility $t(\mathcal M)$. Projective simulation with visibility $t$ means that for all $ i \in \{0,\dots,n-1\}$,
\begin{equation}\label{eq:anc}
\begin{aligned}
\text{Tr}\left(\rho\, \Phi_t(M_i) \right) &= \text{Tr}(\rho\otimes \ketbra{0}{0}_A \sum_{k=1}^N p_k P^k_i)\\
&= \text{Tr}\left(\rho\,\text{Tr}_A\left(\sum_{k=1}^N p_k P^k_i\,\mathbb I\otimes \ketbra{0}{0}_A  \right)\right),
\end{aligned}
\end{equation}where $\Phi_t(M_i) = t M_i + (1-t)\frac{\text{Tr}(M_i)}{d}\mathbb I$, and we have taken the system-ancilla state to be $\rho\otimes \ketbra{0}{0}_A$, without loss of generality. Eq. \eqref{eq:anc} implies that 
\begin{align}\label{eq:anc1}
    \Phi_t(M_i) = \text{Tr}_A\left(\sum_{k=1}^N p_k P^k_i\,\mathbb I\otimes \ketbra{0}{0}_A  \right).
\end{align} We will use this condition to build criteria for projective simulability. Obviously, if $d_A\geq d$, any POVM can be perfectly projectively simulated ($t=1$), which is a consequence of the Naimark dilation theorem. However, if the experimenter has control over a smaller ancilla ($d_A<d$), the POVM can already be simulated to a considerably improved degree compared to no ancilla. For example, with a qubit ancilla, a dimension-independent lower bound on projective simulability of $1/8$ is known \cite{Kotowski2025}; however, generic upper bounds in such situations are not. To investigate this question, we reformulate the problem of projective simulability. 

It has been noted in Refs.~\cite{Khandelwal_2025,Cobucci2025} that a necessary condition for projective simulability can be built by partitioning the space of the projective measurement by the ranks of the measurement operators. We extend this method to investigate the simulation of POVMs with an ancilla. Let us denote as $\vec r =(r_0,r_1,\dots,r_{n-1})$ the tuple such that $r_i$ are non-negative integers and $\sum_{i=0}^{n-1}r_i=d_A d$, where $r_i = \text{rank}(P_i^k)$. We call $\vec r$ a rank tuple and use it to partition the space of projective measurements. We therefore substitute the index $k$ in the above equations by the pair $(\vec r,\mu)$, where $\mu$ is the index of the projective measurement within the set of measurements with rank structure described by $\vec{r}$. Using Eq. \eqref{eq:anc1}, this leads to the condition $M^{\text{h}}_i = \text{Tr}_A\left(\sum_{\vec r} F_i^{\vec r}\,\mathbb I\otimes \ketbra{0}{0}_A  \right)$, where $F_i^{\vec r}=\sum_\mu p_{\vec r,\mu}P^{\vec r,\mu}_i$ and $F_i^{\vec r}\geq 0$. Measurement completeness and normalization lead to $\sum_i F_i^{\vec r} = p_{\vec r}\mathbb I$, where $p_{\vec r} = \sum_\mu p_{\vec r,\mu}$. By definition, the trace of the projectors $\text{Tr}(P^{\vec r,\mu}_i) =r_i$, which implies that $\text{Tr}(F^{\vec r}_i) = p_{\vec r}\,r_i$. We then phrase the question of projective simulation with an ancilla as the following SDP 
\begin{equation}\label{eq:ancSDP}
\begin{aligned}
  \max_{F_i^{\vec r}} & \quad t \\
\text{s.t.}& \,\,	 \Tr_1\Big(\sum_{\vec r} F_i^{\vec r}\,\mathbb I\otimes \ketbra{0}{0}_A  \Big) = \Phi_t(M_i), \,  \text{Tr}(F_i^{\vec r})=p_{\vec{r}}\ r_i,\\
&	\quad\sum_i F_{i}^{\vec{r}}=p_{\vec{r}}\,\mathbb I, \quad\text{and}\quad F_{i}^{\vec{r}}\geq 0.
\end{aligned}
\end{equation}The second constraint implies $p_{\vec{r}}=\frac{1}{d_Ad}\sum_{i=0}^n \tr(F_{i}^{\vec{r}})$. For a given target POVM $\mathcal M$ with effects $\{M_i\}_{i=0}^{n-1}$ and given ancilla dimension, the above SDP gives an upper bound on the maximum visibility $t(\mathcal M)$.

In Tab. \ref{table_herald}, we give the upper bounds on the maximum visibility obtained with the above SDP, for various POVMs considered in this work.

\begin{table}
    \begin{tabular}{ l |c| c | c  }
      $d$ & $d_A$ & POVM $\mathcal{M}$ &$t(\mathcal M)$ (upper bound) \\
      \hline
      3 &2 & flag SIC2  & 1 \\ 
      \hline
      3 &2 &$\text{SIC3a}$  &0.9553  \\ 
      \hline
      3 &2 &$\text{SIC3b}$  & 0.9519  \\ 
      \hline
       3 &2 &$\text{SIC3c}$  & 0.9571 \\ 
      \hline
      4 &2 & flag SIC2 &1  \\ 
      \hline
    4 &2 & flag SIC3a &0.9652   \\ 
          \hline
       4 &2 & flag SIC3b &0.9652  \\ 
      \hline
        4 &2 & flag SIC3c &0.9652  \\ 
      \hline
           4 &2 & SIC4a &?  \\ 
      \hline

    \end{tabular}
    \caption{ SDP results for upper bounds on the critical visibility for various POVMs along with qubit ancilla, obtained using Eq.~\eqref{eq:ancSDP}. The names for SIC-POVMs are according to Ref.~\cite{Scott_2010}.}\label{table_herald}
\end{table}

\section{Outlook and Conclusion}

In this paper, we introduced a hierarchy of semidefinite programs to calculate robust upper bounds on measures of non-simulability of POVMs and showed that in many cases, the calculated bounds on the measure of critical visibility are tight and outperform those of previous bounding methods.

We then exploited the hierarchy to construct witnesses that are robust to both measurement and state preparation errors. Consequently, we implemented two POVMs, a two-dimensional SIC-POVM and a three-dimensional six-effect POVM, on a trapped-ion quantum processor using their projective Naimark extensions to a higher-dimensional space. Using our witness method, we showed that these measurements could not have been simulated using projective measurements in the original Hilbert space, thereby certifying the experiment's use of higher dimensional spaces.

Finally, we exemplified how our results can be extended beyond simulation within the same Hilbert space. To that end, we considered the concept of measurements with ancillas, which simulate general POVMs using projective measurements on a joint system with a dimensionally restricted ancilla. We calculated upper bounds on the critical visibility of such schemes for various POVMs with a qubit ancilla.

While our SDP hierarchy increases in size quickly and numerical evaluations are limited to low levels, we observe that already the lowest nontrivial level yields robust witnesses that outperform previous criteria and can be calculated even for large dimensions.

Apart from providing essential tools for certifying non-projectivity of measurements in experiments, our work uncovered fundamental related questions. For instance, we conjecture that our hierarchy converges to the true value of critical visibility after a finite number of steps. Proving this conjecture could shed light not only on the geometry of the set of simulable POVMs itself, but also on its relation to the occurring subspace of mutually orthogonal marginals. 
Another immediate question concerns the implications of measurement dimensionality. Motivated by our upper bounds on the critical visibility and previous findings that even a two-dimensional ancilla leads to a dimension-independent lower bound on the success probability, it seems that further analyzing the geometry of this landscape might reveal intricate features of quantum theory in general.

\acknowledgments
We thank an anonymous referee for pointing us to results in the PhD thesis of Matthew Pusey that allowed us to prove completeness of the hierarchy. 
This research was funded by the European Union under the Horizon Europe Programme—Grant Agreement 101080086—NeQST and by the European Research Council (ERC, QUDITS, 101039522). Views and opinions expressed are however those of the author(s) only and do not necessarily reflect those of the European Union or the European Research Council Executive Agency. Neither the European Union nor the granting authority can be held responsible for them. We also acknowledge support by the Austrian Science Fund (FWF) through the EU-QUANTERA project TNiSQ (N-6001), by the Austrian Federal Ministry of Education, Science and Research via the Austrian Research Promotion Agency (FFG) through the projects FO999914030 (MUSIQ) and FO999921407 (HDcode) funded by the European Union-NextGenerationEU, and by the IQI GmbH,
A.T.~is supported by the Swedish Research Council under Contract No. 2023-03498 and the Knut and Alice Wallenberg Foundation through the Wallenberg Center for Quantum Technology (WACQT). S.K.~acknowledges support from the Swiss National Science Foundation Grant No. P500PT-222265.
R.B., H.K.~and D.B.~acknowledge support by Deutsche Forschungsgemeinschaft (DFG, German Research Foundation)
under Germany’s Excellence Strategy -- Cluster of Excellence Matter and Light for Quantum Computing (ML4Q) EXC 2004/1 -- 390534769. N.W.~acknowledges support by EIN Quantum NRW.

\appendix
\onecolumngrid
\newpage
\section{Primal and Dual SDP Hierarchy}\label{app:primal_dual_sdp}
The $m$'th level primal SDP $H_m(\mathcal{M})$ is defined as
\begin{align}
  H_m(\mathcal{M}) &= \max_{R_{i_1\dots i_m},t} \quad t \label{sdp_multi_copy_primal_t(d)_app}  \\
  \textrm{s.t.} \quad & R_{i_1\dots i_m} \geq 0, \label{sdp_multi_1_app} \\
  &\sum_{\substack{i_1,\ldots,i_{j-1},\\i_{j+1},\ldots,i_m=0}}^{n-1} \!\!\!\!\!R_{i_1\dots i_m} = \mathbb{I} \otimes \dots \otimes  \Phi_t(M_{i_j}) \otimes \dots \otimes \mathbb{I}, \quad \forall j=1\ldots m, i_j = 0\ldots n-1, \label{sdp_multi_sum_app}  \\
  & \Tr_a (V_{(ab)} R_{i_1\dots i_m}) = 0, \; \forall \; V_{(ab)} \in \mathcal{V}_m \quad \text{if} \; V_{(ab)}\ket{i_1,\dots,i_m} \neq \ket{i_1,\dots,i_m} .  \label{sdp_multi_2_app}
\end{align}
$\mathcal{V}_m$ is the subset of all transpositions. All swap constraints containing larger permutations are automatically fulfilled then. To show this, consider a general permutation $V_{\pi}$ that changes index $a_2$ (we call $V_{\pi}$ the permutation, although it is strictly speaking the representation of the group element $\pi \in S_m$). Then
\begin{align}
    \Tr_{a_2} (V_{\pi} R_{i_1\dots i_m}) = 0.
\end{align}
$V_{\pi}$ can be written as a product of disjoint cycles $V_{\pi} =V_{(\dots)}\dots V_{(a_1a_2\dots a_l)}$. Since every cycle before $V_{(a_1a_2\dots a_l)}$ does not act on system $a_2$, we have 
\begin{align}
    \Tr_{a_2} (V_{(a_1a_2\dots a_l)} R_{i_1\dots i_m}) = 0.
\end{align}
Furthermore we can write any cycle as a product of transpositions $V_{(a_1a_2\dots a_l)} = V_{(a_1 a_l)}V_{(a_1 a_{l-1})} \dots V_{(a_1 a_2)}$. Again, because only $V_{(a_1 a_2)}$ acts on system $k_2$, we finally get
\begin{align}
    \Tr_{a_2} (V_{(a_1 a_2)} R_{i_1\dots i_m}) = 0.
\end{align}

Using standard arguments, the dual program $\mathcal{D}_m(\mathcal{M})$ of the $m$'th level SDP can be calculated to read
\begin{align}\label{sdp_dual_multi_t(d)}
    \mathcal{D}_m(\mathcal{M}) &=\min_{A^{i_\alpha,\alpha},D_{(ab),(i_1,\dots,i_m)}} \quad \sum_{i=0}^{n-1} \left[ \sum_{\alpha=1}^m \Tr(A^{i,\alpha})  \right] \frac{\Tr(M_i)}{d}, \\
    \textrm{s.t.} \quad & \sum_{\alpha=1}^m A^{i_\alpha,\alpha} +\sum_{ \substack{V_{(ab)} \in \mathcal{V}_m,\\ V_{(ab)}\ket{i_1,\dots,i_m} \neq \ket{i_1,\dots,i_m}}}\left[ (\mathbb{I}_d^b \otimes D_{(ab),(i_1,\dots,i_m)}^{\neq b} )V_{(ab)} + \text{h.c.}  \right] \geq 0 \; \forall \; i_1,\dots,i_m \label{eq:second_sum} ,\\
    & \sum_{i=0}^{n-1} \sum_{\alpha=1}^m \left[ \Tr(A^{i,\alpha}) \frac{\Tr(M_i)}{d} - \Tr(A^{i,\alpha}_{ \alpha} M_{i})  \right] \geq 1. \label{eq:third_sum}
\end{align}
$A^{i_\alpha,\alpha}$ are hermitian and $D_{(ab)}$ general complex matrices. For a fixed set of indices $(i_1,\dots,i_m)$, the second sum in \eqref{eq:second_sum} runs over all transpositions that change the $b$'th index. The notation $\mathbb{I}^b$ and $D_{(ab)}^{\neq b}$ means that the operators act on the $b$'th subsystem or everything but the $b$'th subsystem, respectively ($\mathbb{I}$ is a $d \times d$ and $D_{(ab)}^{\neq b}$ a $d^{m-1} \times d^{m-1}$ matrix). The expression $A^{i,\alpha}_{ \alpha}$ denotes the partial trace of $A^{i,\alpha}$ over every system except the $\alpha$'th one.
The last constraint in Eq.~\eqref{eq:third_sum} can be replaced with an equality. To see that, note that taking the trace of Eq.~\eqref{eq:second_sum} for $i_1 = i_2 = \ldots = i_m$ yields $\Tr(A^{i,\alpha}) \geq 0$, implying that the optimization function is always non-negative. Thus, replacing all $A^{i,\alpha}$ by $c A^{i,\alpha}$ (and likewise for the $D_{(ab)}$ matrices) for some constant $c\in(0,1]$ can only yield a better solution, as long as they still obey constraint \eqref{eq:third_sum}. Thus, choosing $c$ as small as possible to reach equality in that constraint is optimal. This observation, together with setting $m=2$ then yields the dual program $\mathcal{D}_2(\mathcal{M})$ that is displayed in Eq.~\ref{sdp_dual_t(d)} in the main text.

For Weyl-Heisenberg covariant SIC-POVMs, we further find the symmetry
\begin{align}
    A^{i,\alpha} = h_i^{\otimes m} A^{0,\alpha} (h_i^\dagger)^{\otimes m}
\end{align}
can be imposed to reduce the number of optimization variables.

\section{Proof of Theorem~\ref{thm:completeness}}\label{app:proof_completeness}

Here, we prove completeness of the hierarchy in Eq.~\eqref{sdp_multi_copy_primal_t(d)} together with the constraints in Eqs.~\eqref{eq:constraint_sym} and \eqref{eq:constraint_sum}. To that end, we make use of Theorem 5.4 from Ref.~\cite{pusey2013quantum}, which establishes a de~Finetti theorem for POVMs. In particular, it states that existence of matrices $R_{i_1\ldots i_m}$ for all $m$ fulfilling constraints \eqref{eq:constraint_sym} and \eqref{eq:constraint_sum}, where the $\tilde{R}_{i_1\ldots i_{m-1}}$ are feasible solutions of the hierarchy level $m-1$ (which, in our case, is established with the same arguments as in the proof of Theorem~\ref{thm:theorem_1}), implies the existence of a measure $\mu(\{E_0,\ldots,E_{n-1}\})$  over the set of $n$-outcome POVMs such that
\begin{align}
    R_{i_1\ldots i_m} = \int \text{d}\mu(\{E_0,\ldots,E_{n-1}\})  E_{i_1} \otimes \ldots \otimes E_{i_m}.
\end{align}
It remains to show that this indeed yields a projective decomposition. To that end, we consider the second level of the hierarchy and evaluate the trace over the  swap constraint in Eq.~\eqref{eq:swap_constraint} for fixed $i < j$:
\begin{align}
    0 = \Tr (VR_{ij} ) = \int \text{d}\mu(\{E_0,\ldots,E_{n-1}\})  \Tr(E_i E_j)
\end{align}
Thus, being a vanishing convex combination of non-negative numbers, almost all POVMs in the decomposition must have orthogonal effects. As $\Tr(E_iE_j) = 0$ implies $E_iE_j = 0$  for positive semidefinite matrices, it follows that almost all POVMs must be projective measurements, yielding a projective simulation of the POVM and the claim follows.

\section{Proof of Lemma \ref{lemma}}\label{app:proof_lemma}

Let $M = \{M_i\}_{i=0}^{n-1}$ be a POVM in dimension $d$ with rank-1 effects, $\Tr(M_i) = \frac{d}{n}$ and $M_iM_j \neq 0$ (for example a SIC-POVM). Then define the $d\cdot m$-dimensional POVM $\tilde{M}$ with $\tilde{M}_i = M_i \otimes \mathbb{I}_m $, where $m$ is chosen, such that $m\frac{d}{n} \in \mathbb{N}$ (e.g. $m=n$ or $m=\frac{n}{\gcd{(d,n)}}$).\\
We first prove i), that $H_2$ detects non-simulability of $\tilde{M}$.

From the constraints \eqref{eq:sum_i_R_ij} and \eqref{eq:sum_j_R_ij} one derives
\begin{align}\label{eq:Rij_app}
    \sum_j R_{ij} = M_i^{A_1} \otimes\mathbb{I}^{A_2}_m \otimes \mathbb{I}^{B_1B_2}_{dm} \quad, \sum_i R_{ij} = \mathbb{I}^{A_1A_2}_{dm} \otimes M_j^{B_1} \otimes \mathbb{I}_m^{B_2} \implies R_{ij} = M_i^{A_1}\otimes M_j^{B_1} \otimes N_{ij}^{A_2B_2},
\end{align}
with $N_{ij}^{A_2B_2} \geq 0$ since the $M_i$ are rank-1. From \eqref{eq:swap_constraint} and $\Tr_{A_1 A_2} (V R_{ii}) = \tilde{M}_i$ (which easily follows from equations \eqref{eq:sum_i_R_ij} and \eqref{eq:swap_constraint}) we get
\begin{align}
    \Tr_{A_1A_2} (R_{ij} V_{A_1B_1} \otimes V_{A_2B_2} ) = \underbrace{M_i^{A_1}M_j^{B_1}}_{\neq 0} \otimes \Tr_{A_2} (V_{A_2B_2}N_{ij}^{A_2B_2}) \overset{!}{=} \delta_{ij} \; M_i^{B_1} \otimes \mathbb{I}^{B_2}_m
\end{align}
From this we can deduce 
\begin{align}
\Tr_{A_2} (V_{A_2B_2}N_{ij}^{A_2B_2}) = 
\begin{cases}
    0, \quad \quad &i\neq j\\
    \frac{n}{d} \mathbb{I}_m, \quad &i = j.
\end{cases}
\end{align}
Then, by multiplying the last equation in \eqref{eq:Rij_app} with $V_{A_2B_2}$ and tracing out $A_2$ we get
\begin{align}
    \Tr_{A_2} (V_{A_2B_2}R_{ij}) = 
\begin{cases}
    0, \quad \quad &i\neq j\\
    \frac{n}{d} M_i^{A_1} \otimes M_j^{B_1} \otimes \mathbb{I}^{B_2}_m, \quad &i = j,
\end{cases}
\end{align}
and therefore
\begin{align}\label{eq:MMI}
    \sum_j \Tr_{A_2} (V_{A_2B_2}R_{ij}) =  \frac{n}{d} M_i^{A_1} \otimes M_j^{B_1} \otimes \mathbb{I}^{B_2}_m.
\end{align}
But from the first equation in \eqref{eq:Rij_app} we also know that
\begin{align}
    \sum_j \Tr_{A_2} (V_{A_2B_2}R_{ij}) = M_i^{A_1} \otimes\mathbb{I}_m^{B_1} \otimes \mathbb{I}^{B_2}_m,
\end{align}
which is in contradiction to \eqref{eq:MMI}, and therefore the constraints in $H_2$ can not be fulfilled simultaneously. 

We now turn to showing ii). The main idea of the necessary conditions on projective simulability in \cite{Oszmaniec2017} and \cite{Cobucci2025} are, to check whether a convex combination into POVMs $N$ exist, where each effect must have integer trace and not more than $d$ effects are allowed. But the POVM $\tilde{M}$ by itself already fulfills these conditions, therefore it finds a feasible solution and detects no non-simulability.

\section{Decomposition from primal solution}\label{app:decomp_primal}
In many cases, we can read off the decomposition in projective measurements from the solution of the primal SDP. The SDP constraints are necessary conditions on
\begin{align}\label{eq:R_i1im_app}
    R_{i_1\dots i_m} = \sum_{k=1}^{N}p_k ( P_{i_1}^k \otimes \dots \otimes P_{i_m}^k).
\end{align}
As an example, consider the decomposition of the noisy qubit SIC-POVM $\mathcal{M^\prime}= \Phi_{\sqrt{2/3}}(\text{SIC}_2)$, which was already given in Ref.~\cite{Oszmaniec2017}. The POVM decomposes into von~Neumann measurements. Therefore, each measurement has two non-vanishing effects spread across four elements, with $\binom{4}{2}=6$ measurements in total, i.e.
\begin{align}
    \begin{pmatrix}
        M_1^\prime\\
        M_1^\prime\\
        M_2^\prime\\
        M_3^\prime
    \end{pmatrix}=
    p_{(0,1)}
    \begin{pmatrix}
        P_0^{(0,1)}\\
        P_1^{(0,1)}\\
        0\\
        0
    \end{pmatrix}+
    \dots +
     p_{(2,3)}
    \begin{pmatrix}
        0\\
        0\\
        P_2^{(2,3)}\\
        P_3^{(2,3)}
    \end{pmatrix}.
\end{align}
But in this specific form, $R_{ij}$ for $i \neq j$ actually becomes a product state, namely $R_{ij} = p_{(i,j)} P_i^{(i,j)} \otimes P_j^{(i,j)}$, because there is only one measurement which has two non vanishing elements at positions $i$ and $j$, and therefore only one term in the sum \eqref{eq:R_i1im_app} survives. This means we can easily read off the projectors and $p_k$ by taking partial traces.

We can do the same for the $3-$dimensional POVMs, the only difference being that we have to go to the third level of the SDP hierarchy. We propose a similar decomposition into $\binom{9}{3}=84$ rank-one projective measurements
\begin{align}
    \begin{pmatrix}
        M_0^\prime\\
        \vdots\\
        M_8^\prime
    \end{pmatrix}=
    p_{(0,1,2)}
    \begin{pmatrix}
        P_0^{(0,1,2)}\\
        P_1^{(0,1,2)}\\
        P_2^{(0,1,2)}\\
        \vdots\\
        0
    \end{pmatrix}+
    \dots +
     p_{(6,7,8)}
    \begin{pmatrix}
        0\\
        \vdots\\
        P_6^{(6,7,8)}\\
        P_7^{(6,7,8)}\\
        P_8^{(6,7,8)}
    \end{pmatrix}.
\end{align}
The projectors can then be read off from $R_{ijk}$ with $ i\neq j, i\neq k, j \neq k$. In the cases we studied, only a subset of all $84$ combinations are non-vanishing. For example, the optimal noisy versions of SIC3a and SIC3b decompose into $63$ projective measurements. An analytical expression of the decomposition in $72$ projective measurements of the noisy version of SIC3c with $t = [1+4\cos(\pi/9)]/6$ can be found in Ref.~\cite{Cobucci2025}.

\section{Proof of Theorem~\ref{thm:witnesspositivity}}\label{app:proof_theorem}
To prove Theorem~\ref{thm:witnesspositivity}, it suffices to consider a simplified SDP of $\mathcal{H}_2(\mathcal{M})$ in \eqref{sdp_multi_copy_primal_t(d)_app}, which fixes the noise parameter to $t=1$
\begin{align}
  H^\ast_2(\mathcal{M}) &= \max_{R_{ij}} \quad 0 \label{eq:simplified_H2}  \\
  \textrm{s.t.} \quad & R_{ij} \geq 0, \label{eq:simpliefied_H2_1} \\
  &\sum_{i}^{n-1} R_{ij} = \mathbb{I} \otimes M_{j}, \quad \sum_{j}^{n-1} R_{ij} = M_i \otimes \mathbb{I}, \quad \forall \; i,j = 0\ldots n-1, \\
  &  \Tr_1 (V R_{ij}) = 0, \quad \text{if} \quad i\neq j. \label{eq:simpliefied_H2_2}
\end{align}
The dual of \eqref{eq:simplified_H2} reads
\begin{align}
\mathcal{D}^\ast_2(\mathcal{M})=& \min_{A^i,D^{ij}} \quad  \Tr(\Tr_2 (A^i) M_i) \\ 
\textrm{s.t.} \quad & A^i + V A^j V +  (\mathbb{I}\otimes D^{ij} V + \text{h.c.}) \geq 0, \quad \forall \; i,j = 0\ldots n-1  \;, \label{sdp_dual_pos_app} \\
&\text{with} \; D^{ij}=(D^{ji})^\dagger, D^{ii} = 0 , \quad \forall \; i,j = 0\ldots n-1. \nonumber   
\end{align}
We can now use the weak duality of semidefinite programming, which  states that for any feasible solution $\Tr(\Tr_2 (A^i) M_i) \geq 0$ holds. In this case, feasible means that the operators $A^i$ fulfill Eq.~\eqref{sdp_dual_pos_app} and Eqs.~\eqref{eq:simpliefied_H2_1}-\eqref{eq:simpliefied_H2_2} are fulfilled as well. Since Eqs.~\eqref{eq:simpliefied_H2_1}-\eqref{eq:simpliefied_H2_2} are just necessary conditions on simulability, this immediately proves Theorem~\ref{thm:witnesspositivity}.

\section{Certifying non-projectivity from experimental data}\label{app:certify_exp}
By analyzing measurement statistics and employing necessary conditions on projective simulability, we can certify that a performed measurement can not be produced by only applying projective measurements.

For this, consider a measurement $\{M_i\}_{i=0}^{n-1}$ on an ensemble of states $\{\rho_\alpha\}_{\alpha=1}^{r-1}$, yielding the probabilities
\begin{align}
    &p_{\alpha,i} = \Tr(\rho_\alpha M_i), \; \alpha \in \{0,\dots,r-1\}, i \in \{0,\dots,n-1\}, \\
    & 0 \leq p_{\alpha,i} \leq 1, \; \forall \; \alpha \in \{0,\dots,r-1\}, i \in \{0,\dots,n-1\}, \\
    &\sum_i p_{\alpha,i} = 1, \; \forall \; \alpha \in \{0,\dots,r-1\}.
\end{align}
The measured probabilities $p^{\text{exp}}_{\alpha,i}$ have some statistical error $\delta_{\alpha,i}$. For simplicity (and due to assuming a constant number of shots for each state measured and applying Hoeffding's bound, see below) we set all statistical errors equal $\delta_{\alpha,i} = \delta$, although the same analysis can be done for different errors. Furthermore we know, the states $\rho_\alpha$ are close to some known target states $\ketbra{\lambda_\alpha}{\lambda_\alpha}$, i.e.
\begin{align}
    F(\rho_\alpha,\ketbra{\lambda_\alpha}{\lambda_\alpha}) = \Tr(\rho_\alpha \ketbra{\lambda_\alpha}{\lambda_\alpha}) \geq F_\alpha.
\end{align}
We want to check, if the measured probabilities are in agreement with probability distributions arising from a projectively simulable POVM. For this, note that we can always define the following matrices for simulable POVMs (similar to SDP \eqref{sdp_multi_copy_primal_t(d)})
\begin{align}
    R_{ij\alpha} = \sum_\lambda p_\lambda  P_{i}^\lambda \otimes P_{j}^\lambda \otimes \rho_\alpha.
\end{align}

One can check the following feasibility SDP
\begin{align}
\text{find} \; &R_{ij\alpha},M_i,\rho_\alpha, q_{\alpha,i}, \label{sdp_cert}\\
\textrm{s.t.} \quad & R_{ij\alpha} \geq 0, R_{ij\alpha}^\Gamma \geq 0,\label{sdp_cert_1}  \\
 & \Tr_1 (V_{(12)} R_{ij\alpha}) = 0 \; \forall \; \alpha, i \neq j, \\
  & \rho_\alpha \geq 0, \Tr(\rho_\alpha) = 1, \Tr(\rho_\alpha \ketbra{\lambda_\alpha}{\lambda_\alpha}) \geq F_\alpha,  \; \forall \; \alpha \in \{0,\dots,r-1\},\\
  & \sum M_i = \mathbb{I}, \; M_i \geq 0, \; \forall \; i \in \{0,\dots,n-1\},\\
  & p_{\alpha,i} - \delta \leq q_{\alpha,i} \leq p_{\alpha,i} + \delta, \; \forall \; \alpha \in \{0,\dots,r-1\}, k \in \{0,\dots,n-1\},  \\
  & \sum_i q_{\alpha,i} = 1, \; \forall \; \alpha \in \{0,\dots,r-1\}, \\
  & \Tr_3 \sum_{i} R_{ij\alpha} = \mathbb{I} \otimes M_j \; \forall \; j, \label{sdp_cert_2} \\
  & \Tr_3 \sum_{j} R_{ij\alpha} = M_i \otimes \mathbb{I} \; \forall \; i , \\
  & \Tr_2 \sum_{i} R_{ij\alpha} \geq p^{\text{max}}_{j} \, \mathbb{I} \otimes \rho_\alpha \, \forall \alpha, j,  \\
  & \Tr_1 \sum_{j} R_{ij\alpha} \geq p^{\text{max}}_{i} \, \mathbb{I} \otimes \rho_\alpha \, \forall \alpha, i,  \\
  & \sum_{i,j} R_{ij\alpha} = \mathbb{I} \otimes \mathbb{I} \otimes \rho_\alpha, \; \forall \alpha \in \{0,\dots,r-1\} ,\\
  &  \Tr_{23} \left( V_{(23)} \sum_{i} R_{ij\alpha} \right) = q_{\alpha,j} \mathbb{I}, \; \forall \alpha, j, \\
  &  \Tr_{13} \left( V_{(13)} \sum_{j} R_{ij\alpha} \right) = q_{\alpha,i} \mathbb{I}, \; \forall \alpha, i \label{sdp_cert_3}
\end{align}
The optimization in \eqref{sdp_cert} is over POVM elements $M_i$, states $\rho_\alpha$ and probabilities $q_{\alpha,i}$, which are related to additional optimization variables $R_{ij\alpha}$ via the constraints (\ref{sdp_cert_2}-\ref{sdp_cert_3}). As an input it takes the experimentally measured probabilities $p_{\alpha,i}$, their statistical errors $\delta$ and lower bounds on fidelities $F_\alpha$. $p^{\text{max}}_{i}$ is the maximal probability measured for the outcome $i$. If the SDP fails to find a solution and is therefore infeasible, this implies that the implemented POVM was not projectively simulable, as the necessary conditions (\ref{sdp_cert_1}--\ref{sdp_cert_3}) cannot be satisfied simultaneously.

The statistical error can be fixed beforehand, but mathematically it is equivalent to minimize this error $\delta$ in \eqref{sdp_cert} and from this deduce a confidence interval via Hoeffding's inequality \cite{Hoeffding01031963}

\begin{align}
    \mathbb{P}\left( |p_{\alpha,i} - \mathbb{E}[\overline{p_{\alpha,i}}]| \geq \delta \right) \leq 2 \exp\left( -2 N \delta^2 \right).
\end{align}

For the 2-dimensional 4-effect POVM we get $\delta \approx 0.0336$.
For the 3-dimensional 6-effect POVM we get $\delta \approx 0.0232$. Both errors correspond to confidence intervals of more than $5\sigma$.

\section{Experimental Implementation}\label{app:experimental_imp}
\subsection{Setup}
\label{app:setup}
The experiment is performed on the same trapped ion qudit quantum processor described in \cite{schindler_quantum_processor, ringbauer_universal_2022}. One $^{40}\text{Ca}^+$ ion is trapped in a linear Paul trap, and information is encoded in the population of the metastable state ($\text{D}_{5/2}$) and the ground state ($\text{S}_{1/2}$). A schematic of both the physical level scheme and the trap geometry can be seen in Fig.~\ref{fig:trap}.  
The ion is  ground-state-cooled using first Doppler cooling\cite{HANSCH197568,wineland_doppler_cooling,laser_cooling_review} and polarization gradient cooling (PGC)\cite{pgc_dalibard,Joshi_2020} coupling with \SI{397}{\nano \meter} light to the dipole transition, followed by resolved sideband cooling with a \SI{729}{\nano \meter} laser on the quadrupole transition.\cite{diedrich_sideband_cooling,Leibfried_2001} The logical states are then manipulated using an addressed \SI{729}{\nano \meter} laser with a $<$ \SI{1}{\hertz} linewidth.\cite{freund_xcorr} Additionally, a \SI{866}{\nano \meter} and a \SI{854}{\nano \meter} laser are needed for repumping the ion from the D-manifold.
\subsection{Qudit Detection}
State readout also utilizes the \SI{397}{\nano \meter} laser, scattering light resonantly and collecting it using an objective and focusing it onto an Andor EMCCD camera. Scattering light then corresponds to the ion being in the $\text{S}_{1/2}$ manifold, no light to the $\text{D}_{5/2}$ manifold. In order to discriminate multiple different levels of the $\text{D}_{5/2}$ manifold, multiple sequential readout cycles have to be used: If the ion starts with only one state encoded in the $\text{S}_{1/2}$ manifold (e.g. $\ket{0}$, a first readout only gives information on whether the ion was in $\ket{0}$ or another state. If the ion is found dark, a narrow, resonant \SI{729}{\nano \meter} laser pulse maps the $\ket{1}$ state to the S-manifold, followed by another detection. If the ion is again found dark, this repeats with $\ket{2}$, until the qudit state of the ion is fully determined. 
\begin{figure}[h!]
    \hspace{-5em}
    \begin{minipage}{0.4\linewidth}
    \textbf{(a)}\par\vspace{0.8em}
        \includegraphics[height=4cm]{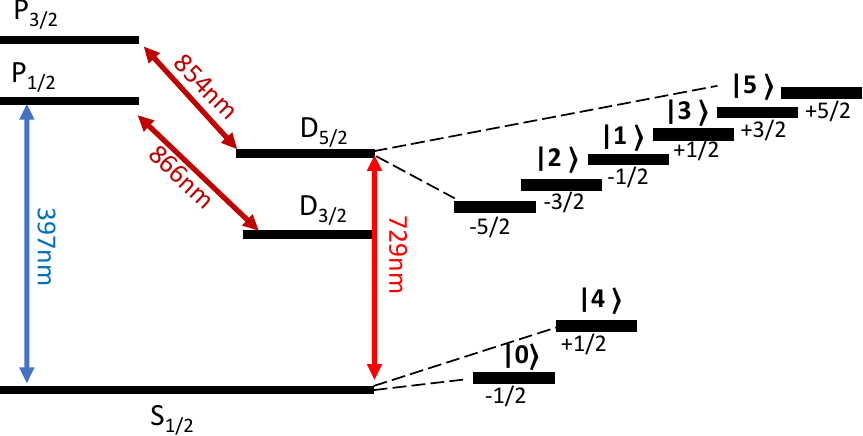}
    \end{minipage}
    \hspace{1cm}
    \begin{minipage}{0.4\linewidth}
    \textbf{(b)}\par\vspace{-0.2em}
    \hfill
    \includegraphics[height=4cm]{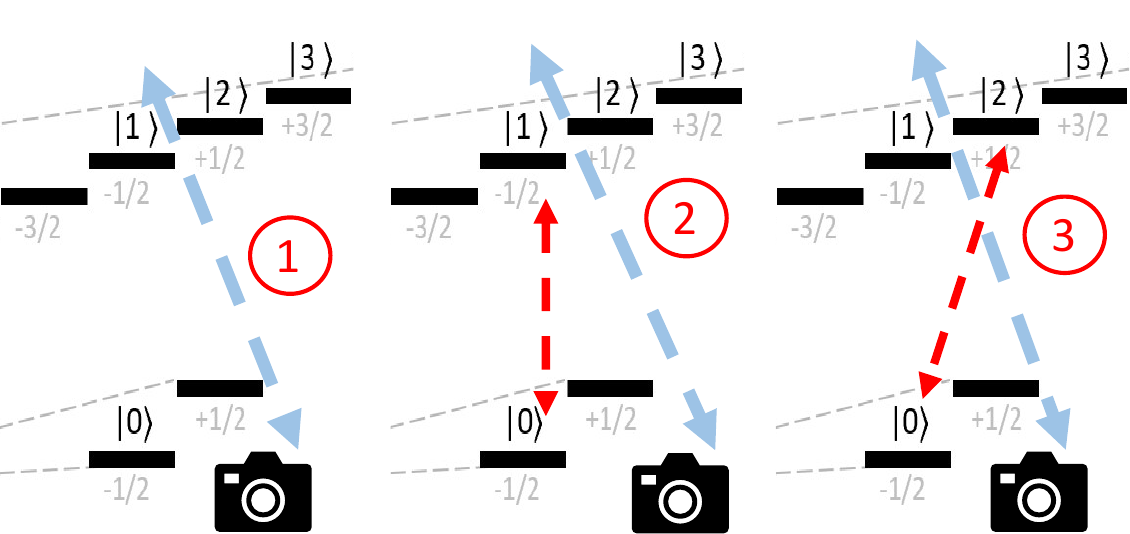}

    \end{minipage}
    \caption{ \textbf{(a)} Schematic level diagram of a $^{40}\text{Ca}^+$ ion, including the Zeeman-splitting in the $\text{S}_{1/2}$ and $\text{D}_{1/2}$ manifold. States used for logic are labeled as $\ket{0}$ etc. Arrows indicate the different lasers that are used for qudit manipulation as well as their wavelengths. 
    \textbf{(b)} Illustration of the qudit readout process, in the case of a ququart. In this example, we have to implement the readout routine three times to fully characterize the ion state.}
    \label{fig:level}
\end{figure}

\subsection{POVM Measurement}\label{app:POVMmeasurement}
Experimentally, two different POVMs are certified for non-simulability: The 4-effect SIC-POVM in a 2d system $\mathcal{M}_2$, and a 6-effect informationally complete POVM in a real 3d system $\mathcal{M}_{3r}$.

We first measured the qubit SIC-POVM $\mathcal{M}_2$, as e.g., used in \cite{stricker_experimental_2022}, given by projectors onto the tetrahedral vectors in Eq.~\eqref{eq:psi_sic2}. The POVM measurement is implemented using unitary $\hat{U}_4$ to map the qubit onto the larger available ququart space (vectors {$\ket{0}-\ket{3}$}). This acts as a Naimark dilation, so that the outcomes from a projective measurement in the computational basis of the ququart directly correspond to the four different POVM outcomes.

The 6-effect POVM $\mathcal{M}_{3r}$is measured in the same way. It consists of projectors onto the vectors in Eq.~\eqref{eq:phi_sicr3}, and is implemented using unitary $\hat{U}_6$ that maps the qutrit to the qusext still available in the $^{40}\text{Ca}^+$ ion, which can then be measured in the computational basis. The corresponding unitaries read
\begingroup
\renewcommand*{\arraystretch}{1.5}

\begin{align*} \hat{U}_4 = 
\begin{pmatrix}
  \frac{1}{\sqrt{2}} & 0 & 0 & \frac{1}{\sqrt{2}} \\
  \frac{1}{\sqrt{6}} & \frac{1}{\sqrt{3}} & \frac{1}{\sqrt{3}} & \frac{-1}{\sqrt{6}} \\
  \frac{1}{\sqrt{6}} & \frac{(-1)^\frac{2}{3}}{\sqrt{3}} & \frac{(-1)^\frac{4}{3}}{\sqrt{3}} & \frac{-1}{\sqrt{6}} \\
  \frac{1}{\sqrt{6}} & \frac{(-1)^\frac{4}{3}}{\sqrt{3}} & \frac{(-1)^\frac{2}{3}}{\sqrt{3}} & \frac{-1}{\sqrt{6}} \\
\end{pmatrix},
\quad
\hat{U}_6 = 
\begin{pmatrix}
  \frac{1}{2} & \frac{1}{2} & 0 & -\frac{1}{2}& 0 & -\frac{1}{2} \\
  \frac{1}{2} & 0 & \frac{1}{2} & 0 & \frac{1}{2} & \frac{1}{2} \\
  0 &\frac{1}{2} & \frac{1}{2} & \frac{1}{2} & -\frac{1}{2}  &  0 \\
  \frac{1}{2} & -\frac{1}{2}  & 0 & \frac{1}{2} &  0&  -\frac{1}{2} \\
  \frac{1}{2} & 0 & -\frac{1}{2}  & 0 &  -\frac{1}{2}  &  \frac{1}{2} \\
  0 &\frac{1}{2} & -\frac{1}{2}  & \frac{1}{2} & \frac{1}{2} &  0 \\
\end{pmatrix}.
\end{align*}

\endgroup

The fidelities of the prepared states are estimated as explained in Sect.~\ref{sec:experiment} in the main text and can be found in Tab.~\ref{tab:fidelities}. The probabilities of the different outcomes are listed in Tab.~\ref{tab:povm_statistics}. While the state certification used 50.000 shots per state, the POVM was measured using 40.000 shots. This is because state certification has a greater influence on the POVM witness uncertainty intervals. 

\begin{table}[t]
    \centering
    \begin{minipage}{0.35\linewidth}
    \begin{tabular}{c|c}
      State   & Fidelity \\ \hline
      $\ket{\Psi_0}$  & 0.9999(1)\\
       $\ket{\Psi_1}$ & 0.9989(2)\\
       $\ket{\Psi_2}$  & 0.9989(2)\\
       $\ket{\Psi_3}$  & 0.9992(2)\\
         
    \end{tabular}
    \end{minipage}
    \begin{minipage}{0.35\linewidth}
    \begin{tabular}{c|c}
       State  & Fidelity \\ \hline
       $\ket{\Phi_0}$  & 0.9965(3) \\
       $\ket{\Phi_1}$ & 0.9975(2) \\
       $\ket{\Phi_2}$ & 0.9971(2) \\
       $\ket{\Phi_3}$ & 0.9970(2) \\
       $\ket{\Phi_4}$ & 0.9967(3) \\
       $\ket{\Phi_5}$ & 0.9967(3) \\
    \end{tabular}
    \end{minipage}
    \label{tab:fidelities}
    \caption{Lower bounds on the state preparation fidelities. The left table contains all states used to certify the 4-effect POVM, the right table the states used to certify the 6-effect POVM. The given fidelities are lower bounds, where readout error is included in the infidelity given.  The number in brackets is the uncertainty on the last shown digit, given as the shot-wise binomial standard deviation of the projective measurement.}
\end{table}

\begin{table}[t]

    \begin{minipage}{0.8\linewidth}
    \centering
    \begin{tabular}{c|c|c|c|c}
      State   & "0" & "1" & "2" & "3" \\ \hline
       $\ket{\Psi_0}$  & 0.498(2) &  0.172(1) & 0.163(1)& 0.167(1)\\
       $\ket{\Psi_1}$ & 0.165(1)& 0.498(2)& 0.172(1) &  0.165(1)\\
       $\ket{\Psi_2}$  & 0.169(1) &  0.162(1) &0.501(2) & 0.168(1)\\
       $\ket{\Psi_3}$  & 0.168(1) &0.168(1) &0.166(1) & 0.499(2)\\
         
    \end{tabular}    
    \newline
    \begin{tabular}{c|c|c|c|c|c|c}
       State &"0"  & "1" & "2" & "3" & "4" & "5" \\ \hline
       $\ket{\Phi_0}$ &0.486(2)& 0.135(2)& 0.125(2) & 0.0017(2) & 0.126(2)& 0.127(2) \\
       $\ket{\Phi_1}$ & 0.142(2)& 0.487(2) & 0.128(2)  & 0.117(2)& 0.0014(2)& 0.124(2) \\
       $\ket{\Phi_2}$ & 0.130(2)& 0.118(2)& 0.499(2)&0.127(2) &0.124(2)&0.0035(3) \\
       $\ket{\Phi_3}$ & 0.0020(2)& 0.139(2)& 0.116(2)& 0.477(2)&0.1320 (2)&0.134(2)\\
       $\ket{\Phi_4}$ & 0.123(2)& 0.0019(2)&0.133(2)& 0.122(2)  &0.490(2)& 0.130(2) \\
       $\ket{\Phi_5}$ &0.112(2)& 0.112(2)& 0.0020(2) & 0.124(2)& 0.121(2)& 0.514(2) \\
    \end{tabular}
    \end{minipage}
    \label{tab:povm_statistics}
    \caption{Qubit and qutrit POVM outcome probabilities.}
\end{table}

\bibliography{apssamp}

\end{document}